\documentclass[9pt,journal]{IEEEtran}
\IEEEoverridecommandlockouts                        
\usepackage{enumerate} 
\usepackage{color}
\usepackage[dvipsnames]{xcolor}
\usepackage{graphicx}
\usepackage{setspace}
\usepackage{bbm}
\usepackage{mathrsfs,mathtools}
\usepackage{amssymb,amsfonts,amsmath,cite,comment,relsize,latexsym}
\usepackage{caption}
\usepackage[dvips]{psfrag}
\usepackage{setspace}
\usepackage{amsmath}
\usepackage{url}
\usepackage{soul}
\usepackage{pgf,tikz}
\usepackage{subfig}
\usepackage{tabularx}
\usepackage{algorithm}
\usepackage[noend]{algpseudocode}
\usepackage{epstopdf}

\newtheorem{theorem}{Theorem}[section]
\newtheorem{lemma}[theorem]{Lemma}

\newtheorem{remark}[theorem]{Remark}
\newtheorem{definition}[theorem]{Definition}
\newtheorem{example}[theorem]{Example}
\newtheorem{assumption}[theorem]{Assumption}

\DeclareMathOperator{\E}{\mathbb{E}}
\DeclareMathOperator{\R}{\mathbb{R}}

\DeclareMathOperator{\Tr}{\text{Tr}}
\DeclareMathOperator{\T}{\text{T}}
\DeclareMathOperator{\diag}{\mathrm{diag}}
\DeclareMathOperator{\Vect}{\mathrm{Vec}}
\DeclareMathOperator{\HH}{\mathcal{H}}
\newcommand{\G}{{\mathcal{G}}}
\newcommand{\V}{{\mathcal{V}}}
\newcommand{\EE}{{\mathcal{E}}}
\DeclareMathOperator{\CT}{\text{H}}
\DeclareMathOperator{\N}{\mathnormal{n}}
\DeclareMathOperator*{\argmin}{arg\,min}
\DeclareMathOperator*{\argmax}{arg\,max}
\newcommand{\rhoo}{\rho_\text{ss}}
\newcommand{\rhoot}{\tilde{\rho}_\text{ss}}
\newcommand{\LK}{L_{_F}}
\newcommand{\Wc}{\varpi}
\newcommand{\wc}{w}

\newcommand{\Sp}{\hspace{0.05cm}}
\newcommand{\numRowsOfC}{m}

\newcommand{\tagarray}{%
	\mbox{}\refstepcounter{equation}%
	$(\theequation)$%
}

\newcommand{\BR}{ }
\newcommand{\BL}{ }
\newcommand{\HS}{\hspace{0.05cm}}

\usepackage[normalem]{ulem}
\newcommand{\stkout}[1]{\ifmmode\text{\sout{\ensuremath{#1}}}\else\sout{#1}\fi}

\makeatletter
\def\BState{\State\hskip-\ALG@thistlm}
\makeatother

\title{\Huge 
	Performance  Improvement in {\BL Noisy} Linear  \\Consensus Networks with Time-Delay}

\author{Yaser Ghaedsharaf$^{\dagger}$, Milad Siami$^{\ddagger}$, Christoforos Somarakis$^{\dagger}$, and Nader Motee$^{\dagger}$ 
	 \thanks{This research is supported by NSF CAREER ECCS-1454022 and ONR YIP N00014-16-1-2645.}
	\thanks{$^{\dagger}$ 
	Y. Ghaedsharaf, C. Somarakis, and N. Motee are with the Department of Mechanical Engineering and Mechanics, Lehigh University, Bethlehem, PA 18015, USA {\tt\small \{ghaedsharaf, csomarak, motee\}@lehigh.edu}}
	\thanks{$^{\ddagger}$M. Siami is with the Institute for Data, Systems, and Society, Massachusetts Institute of
	Technology (MIT), Cambridge, MA 02139, USA.
	{\tt\small siami@mit.edu}}
}

\begin{document}
	
	\maketitle
	\thispagestyle{empty}
	\pagestyle{empty}
	
	\begin{abstract}
		We analyze performance of a class of time-delay first-order consensus networks from a graph topological perspective and present methods to improve it. The performance is measured by network's square of $\mathcal{H}_2$-norm and it is shown that it is a convex function of Laplacian eigenvalues and the coupling weights of the underlying graph of the network. First, we propose a tight convex, but simple, approximation of the performance measure in order to achieve lower complexity in our design problems by eliminating the need for eigen-decomposition. The effect of time-delay  reincarnates itself in the form of non-monotonicity, which results in nonintuitive behaviors of  the performance as a function of graph topology. Next, we present three methods to improve the performance by growing, re-weighting, or sparsifying the underlying graph of the network. It is shown that our suggested algorithms provide near-optimal solutions with lower complexity with respect to existing methods in literature. 
	\end{abstract}

	\section{INTRODUCTION}\label{sec1}
	\allowdisplaybreaks
	
	{\it Our objective} is to characterize ${\HH}_2$-norm performance of a noisy time-delay linear consensus network using the spectrum of the Laplacian matrix, quantify inherent fundamental limits on its best achievable performance, and eventually develop low time-complexity and efficient algorithms to improve the performance.
	
	\vspace{0.1cm}
	{\it Literature Review:} Measures for performance of consensus networks and the problem of designing such networks to achieve optimal performance have been extensively studied in the past decades. Performance of consensus networks in the absence of time-delay were studied in \cite{Siami16tac,Bamieh:2012,young2010robustness,Zelazo:2011}.
	Optimal design of averaging networks was studied in \cite{Xiao:2004}.
	Minimizing the total effective resistance of the graph was investigated in \cite{ghosh2008minimizing}. By using the fact that the total effective resistance of underlying graph is proportional to the $\HH_2$-norm squared of a first-order consensus networks, papers \cite{siami2016topology,summers2015topology,moghaddam2015interior} consider improving the performance measure by growing the underlying graph.
	Authors in \cite{Olfati:2004}  establish the relation between algebraic connectivity of the coupling graph of the network and performance of linear consensus networks in absence of time-delay. NP-hardness of the problem of adding a prespecified number of edges to a network via  maximizing its algebraic connectivity is proven in  \cite{mosk2008maximum}, while \cite{ghosh2008minimizing} suggested a heuristic using the Fiedler vector of the graph to address the problem. Despite extensive study of networks' performance in the absence of time-delay, limited attention has been given to performance analysis  of linear consensus networks in presence of time-delay.
	
	Stability analysis of first-order linear consensus networks with homogeneous time-delay and undirected couplings was studied in \cite{Olfati:2004}.   Necessary and sufficient conditions for stability of linear undirected network with non-uniform delay were reported in \cite{Munz:2010}. Scalable condition for robust stability of networks with heterogeneous stable dynamics using s-hulls was considered in \cite{lestas2006scalable,lestas2007scalable,lestas2007s}. The authors of  \cite{jonsson2007low} propose a low complexity stability criterion for a class of large-scale systems and a scalable robust stability criterion for interconnected systems with heterogeneous linear time-invariant subsystems is reported in \cite{jonsson2010scalable}. Stability analysis of consensus and oscillatory nonlinear networks as well as switching networks with heterogeneous time-delays were tackled in \cite{Papachristodoulou:2005}. The convex hull notion was utilized in \cite{kao2009characterization,jonsson2007low,jonsson2010scalable} to study robust stability of networked systems with normal interconnection matrices. Moreover, a unifying framework to analyze network synchronization using integral quadratic constraints is proposed in \cite{khong2016unifying}. In \cite{somarakis2015delay}, authors study convergence rate of averaging networks subject to time-delay.   The authors of \cite{Qiao:2013} aim at designing a robustly stable network with respect to time-delay, and in\cite{rafiee2010optimal}, the authors maximize algebraic connectivity of underlying  graph of a time-delay linear consensus network; nonetheless, importance and relevance of algebraic connectivity in performance evaluation of time-delay networks remains disputable.   Lastly,  \cite{hunt2012networkThesis, hod2010analytic,hunt2010network} analyze synchronization efficiency in a first-order consensus network with uniform or multiple delays. Although most of their work are limited to numerical results, some of them coincide with our work in derivation of the performance measure  in a special case, i.e., when the output matrix is a centering matrix.  
	
	{\BL The purpose of this manuscript is to analyze performance of time-delay linear consensus networks and propose design algorithms to achieve best attainable performance for such networks. We assume that all time-delays are identical and coupling (communication) graphs of networks are undirected (bidirectional). These assumptions allow us to quantify performance explicitly using closed-form formulae. Relaxing these assumptions to include networks with nonuniform time delays and/or directed coupling graphs are challenging and require separate investigations, which may not lend themselves to analytic performance evaluation and prevent us from devising low time-complexity algorithms \cite{dezfulian2018performance,moradian2017dynamic}. Since time-delay is intrinsic to all networked control systems, devising efficient and  scalable algorithms for analysis and design of networked systems with higher-order dynamics and time-delay will remain an active research area, for many years to come.
	
	The current manuscript extends results of  \cite{ghaedsharaf2016interplay,ghaedsharaf2016complexities,ghaedsharaf2017performance} and presents a consistent story on how to analyze and improve the performance of noisy linear consensus networks subject to time-delay. In \cite{ghaedsharaf2016interplay}, we studied properties of the first-order consensus network's $\HH_2$-norm in the presence of time-delay, and in \cite{ghaedsharaf2016complexities,ghaedsharaf2017performance}, we offered growing and sparsification  algorithms to enhance the $\HH_2$-norm of the network. This manuscript extends  results of \cite{ghaedsharaf2016interplay,ghaedsharaf2016complexities,ghaedsharaf2017performance}  to the general output matrix and provides detailed proofs and explanation for all theorems and lemmas. Furthermore, Theorems \ref{RobustTop}, \ref{RobustTopPath}, \ref{RobustTopRing}, which relate the performance measure to known graph topologies
	are new additions to this manuscripts. Algorithm  \ref{randGreedy}, which is superior to Algorithm \ref{simpGreedy}, is also our new contribution. This superiority is shown through Example \ref{simp-rand}. In addition, materials in section \ref{sec-IX}  are new, which provided a theoretical guideline on which design methods (reweighing, growing, or sparsification) should be utilized to enhance the performance.

	\vspace{0.1cm}
	Our main contributions include (i) investigation of functional  properties of the performance measure and characterization of fundamental limits on its best achievable values, (ii) efficient approximation of the performance measure for network synthesis, (iii) low time complexity  algorithms to design state feedback controllers for  performance enhancement of time-delay linear consensus networks. 
	In section \ref{sec-III}, we express the $\HH_2$-norm performance of a  time-delay linear consensus network in terms of its Laplacian spectrum. Furthermore, we prove that this performance measure is convex with respect to coupling weights and Laplacian  spectrum, and in addition, it is an increasing function of time-delay. 
	In section \ref{sec-IV}, we discuss topologies with optimal performance. Furthermore, we quantify a sharp lower bound on the best achievable performance for a  network with a fixed time-delay. In presence of time-delay, the $\HH_2$-norm performance of first-order consensus network is not monotone decreasing with respect to connectivity, which impose challenges in design of the optimal network as  increasing connectivity may deteriorate the performance. Then, we present methods to improve the performance measure. We categorize these procedures as growing, reweighting, and sparsification. Although the $\HH_2$-norm performance is a convex function of Laplacian eigenvalues, direct use of  this spectral function in our network design problems requires eigen-decomposition, which adds to time complexity of our design procedures. To overcome this, our key idea is to calculate an approximation function of the performance measure that spares us eigen-decomposition of the Laplacian matrix. In section \ref{sec-V}, we tackle the combinatorial problem of improving the non-monotone performance measure of the time-delay network by adding new interconnection links.  Our time-complexity analysis of our proposed algorithm to grow a time-delay network shows that it can be done in $\mathcal{O}(n^3+ \numRowsOfC n^2+kn^2)$ arithmetic operations, where $n$ is number of nodes, $\numRowsOfC$ is number of rows of the output matrix $C$  and $k$ is maximum number of new interconnections.
	Section \ref{sec-VI} discusses reweighting of the coupling weights as an approach to improve the performance measure. This design problem can be cast as a semidefinite programming (SDP) problem, which inherits time-complexity of existing SDP solvers. In the absence of time-delay, removing interconnections will deteriorate the $\HH_2$-norm performance due to monotonicity property of the performance measure. In section \ref{sec-VII}, we explain how one can sparsify the coupling graph of the time-delay network to improve the performance measure. In section \ref{sec-VIII}, we discuss sensitivity of the performance measure with respect to weight of couplings, where it helps us in knowing whether we should improve the performance by increasing connectivity through growing the network or we should sparsify the network without knowing the spectrum of the Laplacian matrix. We use several simulation case studies to show effectiveness of our proposed design algorithms.

	\section{Preliminaries and Definitions}\label{sec2}
	\subsection{Basic Definitions}
	Throughout the paper the following notations will be used. We denote transpose and conjugate transpose of matrix $A$ by $A^{\T}$ and $A^{\CT}$, respectively. Also, set of non-negative (positive) real numbers is indicated by $\mathbb{R}_{+}$ ($\mathbb{R}_{++}$). An undirected weighted graph $\mathcal{G}$ is denoted by the triple $\mathcal{G=(V,E,}w)$,  where $\mathcal{V}=\{v_1,v_2, \dots, v_{\N}\}$ is set of nodes (vertices) of the graph, $\EE$ is set of links (edges) of the graph,  and $w: \EE \rightarrow \mathbb{R}_{++}$ is the weight function that maps each link to a positive scalar. We let $L$ to be the Laplacian of the graph, defined by $L=\Delta-A$, where $\Delta$ is diagonal matrix of node degrees and $A$ is the adjacency matrix of the graph. The $n \times 1$ vector of all zeros and ones  are denoted by $0_{\N}$ and $\mathbf{1}_{\N}$, respectively, while $J_{\N}=\mathbf{1}_{\N}\mathbf{1}_{\N}^{\T}$ is the $n \times n$ matrix of all ones. Furthermore, the $n \times n$ centering matrix is denoted by $M_{\N}=I_{\N}-\frac{1}{\N}J_{\N}$. For an undirected graph with $\N$ nodes, Laplacian eigenvalues are real and  shown in an order sequence as $0=\lambda_1< \lambda_2
	\leq\dots\leq \lambda_{\N}$. We denote the complete unweighted graph by $\mathcal{K}_{\N}$.
	We indicate Moore-Penrose pseudo-inverse of a matrix $P$ by $P^{\dagger} = [p_{ij}^{\dagger}]$ and we define
	\begin{align}
	r_e(P) \coloneqq p_{ii}^{\dagger}+p_{jj}^{\dagger}-2 p_{ij}^{\dagger}
	\end{align}
	for every given link $e=\{i,j\}$. Accordingly, in a graph with Laplacian matrix $L$, the effective resistance between two ends of a given link $e=\{i,j\}$ is denoted by $r_e(L)$. For a given $n$-tuple $y$, operator $\diag(y)$ maps   $y$ to an $n\times n$ diagonal matrix whose main diagonal elements are elements of $y$.  For a matrix $P\in \mathbb{R}^{n \times m}$, the vectorization of $P$, denoted by $\Vect(P) \in \mathbb{R}^{nm}$, is a vector obtained by stacking up columns of matrix $P$ on top of one another. For a square matrix $X$, matrix functions $\cos (X)$ and $\sin(X)$ are defined as 
	\begin{eqnarray}
	\cos(X) = \sum_{k=0}^{\infty}\frac{(-1)^k X^{2k}}{(2k)!},~ \sin(X) = \sum_{k=0}^{\infty}\frac{(-1)^k X^{2k+1}}{(2k+1)!}.\label{sin-cos}
	\end{eqnarray}

	\begin{definition}\label{SchurDef}
		A function $g: \mathbb{R}^{\N}\rightarrow \mathbb{R}$ is Schur-convex  if for every doubly stochastic matrix $D \in \mathbb{R}^{\N \times \N}$ and all ${x \in \mathbb{R}^{\N}}$, we have 
		\begin{align*}
		g(Dx)\leq g(x).
		\end{align*}
	\end{definition}
	
	\subsection{Noisy Consensus Networks with Time-Delay}
	We consider the class of  linear dynamical networks that consist of multiple agents with scalar state variables $x_i$ and control inputs $u_i$ whose dynamics evolve in time according to 	
	\begin{equation*}
	\dot{x}_i(t)  =  u_i(t) +\xi_i(t) \label{TI-consensus-algorithm} 
	\end{equation*}
	for all $i=1,\ldots,n$, where initial condition $x_i(0)=x_i^0$ is given. It is assumed that 
	$x_i(t)= 0$ for all $t \in [-\tau, 0)$. The impact of an uncertain environment on each agent's dynamics is modeled by the exogenous noise input $\xi_i(t)$. 
	We assume that every agent experiences a time-delay in accessing, computing, or sharing its own state information with itself and other neighboring agents. It is assumed that all time-delays for all agents are identical and equal to a nonnegative number $\tau$. We apply the following feedback control law
	\begin{equation}
	u_i(t) ~=~\sum_{j=1}^{n} k_{ij} \big(x_j(t-\tau) - x_i(t-\tau)\big),\label{feedback-law}
	\end{equation}
	to every agent of this network. The resulting closed-loop network will be a first-order linear consensus network, whose dynamics can be written in the following compact form
	\begin{subequations}\label{eq:system}
		\begin{align}
		\dot x(t) & =   -L\, x(t-\tau)~+~\xi(t)\label{first-order}\\
		y(t)& =\,\, C \,x(t)\label{first-order-output}
		\end{align}
	\end{subequations}
	%
	with $x(t)= 0$ for all $t \in [-\tau, 0)$ and $x(0)=x^0$, where  $x^0 = [x^0_1,  \ldots,  x^0_n]^{\rm T}$ is the initial condition, $x = [x_1,  \ldots,  x_n]^{\rm T}$ is the state, $y = [y_1,  \ldots,  y_{\numRowsOfC}]^{\rm T}$ is the output, and $\xi = [\xi_1,  \ldots,  \xi_n]^{\rm T}$ is the exogenous noise input of the network. It is assumed that $\xi(t)$ is a vector of independent Gaussian white noise processes with zero mean and identity covariance, i.e.,
	\begin{align*}
	\E\Big[\xi(t_1)\xi^{\T}(t_2) \Big] = I_{\N} \delta(t_1-t_2),
	\end{align*}
	where $\delta(t)$ is the delta function. 
	The state matrix of the network is a graph Laplacian that is defined by $L=[l_{ij}]$, where 
	\begin{equation*}
	\displaystyle l_{ij} := \left\{\begin{array}{ccc}
	-k_{ij} & \textrm{if} & i \neq j \\
	k_{i1}+\ldots+k_{in}& \textrm{if} & i=j
	\end{array}\right..
	\end{equation*}

	\vspace{0.1cm}
	\begin{assumption}\label{C_ones} The vector of all ones  is in the null space of the output matrix, i.e., $C 1_{\N}=0$.
	\end{assumption}

	The underlying coupling graph of the consensus  network \eqref{first-order}-\eqref{first-order-output} is a graph $\G=(\V,\mathcal E, w)$ with node set $\V=\{1,\ldots,n\}$, edge set 
	\begin{equation*} 
	\EE=\Big\{ \{i,j\}~\big|~\forall~i,j \in \V,~k_{ij} \neq 0\Big\}, 
	\end{equation*}
	and weight function 
	\begin{equation*}
	w(e)=k_{ij}   
	\end{equation*}
	for all $e=\{i,j\} \in \EE$, and $w(e)=0$ if $e \notin \EE$. The Laplacian matrix of graph $\G$ is equal to $L$. 
	\vspace{0.2cm}
	
	\begin{assumption}\label{assump-simple}
		All feedback gains satisfy the following properties for all $i,j \in \V$:  \noindent (i)~non-negativity: $k_{ij} \geq 0$, 
		\noindent (ii)~symmetry: $k_{ij}=k_{ji}$,
		\noindent (iii)~simpleness: $k_{ii}= 0$.
	\end{assumption}
	
	\vspace{0.2cm}
	Property (ii) implies that the underlying graph $\G$ is undirected and property (iii) means that there is no self-loop in the network.
	
	\vspace{0.2cm}
	\begin{assumption}\label{assum-coupling-graph}
		The coupling graph $\G$ of the consensus network \eqref{first-order}-\eqref{first-order-output} is connected and time-invariant.
	\end{assumption}
	\vspace{0.1cm}
	
	This assumption implies that only the smallest Laplacian eigenvalues  is equal to zero, i.e., $\lambda_1=0$ and all other ones are strictly positive, i.e., $\lambda_i > 0$ for $i=2,...,n$.

	\subsection{Network Performance Measures}
	When there is no input noise, i.e., $\xi(t) \equiv 0_n$, it is already known \cite{Olfati:2004} that under the condition
	\begin{equation}\label{eq:stabilityOlfati}
	\lambda_{\N}~<~\frac{\pi}{2\tau} 
	\end{equation}
	as well as graph connectivity, states of all agents converge to average of all initial states; whereas in presence of input noise, the agents' states fluctuate around their average.

	In order to quantify the quality of noise propagation in dynamical network  \eqref{eq:system}, we adopt the following performance measure  
	\begin{equation}\label{eq:coherenceDef}
	\rhoo(L;\tau)
	=\lim_{t \rightarrow \infty} \E\Big[y^{\T}(t) y(t)\Big].
	\end{equation}
	
	It can be verified that transfer function of consensus network (\ref{eq:system}) is equal to the transfer function of the following system
	\begin{align}
	\label{eq:systemProjected}
	\begin{cases}
	\dot{\hat{x}}(t)&=-\big({L}+\frac{1}{\tau \N}J_{\N}\big)\hat{x}(t-\tau)+M_{\N}\xi(t)\\
	y(t)  & =C \hat{x}(t)
	\end{cases},
	\end{align}
	where $\hat{x}(t)=M_{\N}x(t)$ is projection of network's states on to the disagreement subspace, i.e., $\hat{x}_i(t)=x_i(t)-\bar{x}(t)$ in which 
	\[\bar{x}(t)=\frac{1}{n}\big(x_1(t)+\ldots+x_n(t)\big). \]
	Since for $\tau < \frac{\pi}{2\lambda_{\N}}$ the system (\ref{eq:systemProjected}) is exponentially stable and transfer function of system (\ref{eq:systemProjected}) is identical to transfer function of system (\ref{eq:system}), we infer that the single marginally stable mode of the consensus network (\ref{eq:system}) (which corresponds to $\lambda_1=0$) is not observable in the output $y(t)$ according to Assumption \ref{C_ones}, which results in boundedness and well-definedness of the performance measure \eqref{eq:coherenceDef}. This method has  been widely exploited in the literature\cite{siami2017centrality,ghaedsharaf2017eminence,Zelazo:2011,young2010robustness}.
	
	\vspace{0.1cm	}
	{ We now list three different and meaningful 
		coherency measures that has been lately used in the context of linear consensus network  \cite{Bamieh:2012, young2010robustness, Zelazo:2011}.
		
		\vspace{0.1cm}
		\noindent {\it  (i) Pairwise deviation.}   
		\begin{align*}
		\frac{1}{2n}\sum_{i,j=1}^n (x_i - x_j)^2 \, = \,&  x^T(t) B_{\mathcal{K}_{\N}}^{\T} \Sp \mathrm{diag} \left(\frac{1}{n},\ldots,\frac{1}{n} \right) \Sp B_{\mathcal{K}_{\N}} x(t) \\
		\, = \,& x^T(t) M_n x(t)
		\end{align*}
		where $B_{\mathcal{K}_{\N}}$ is the signed edge-to-vertex incidence matrix of the complete graph ${\mathcal{K}_{\N}}$. 
		
		\vspace{0.2cm}
		\noindent {\it  (ii) Deviation from average.} 
		
		\begin{equation*}
		\| x(t) - \bar{x}(t)\mathbf{1}_{n}\|_2^2 = \sum_{i=1}^n \Sp (x_i(t) - \bar{x}(t))^2= x^T(t) M_n x(t) 
		\end{equation*}

		\vspace{0.1cm}
		\noindent {\it  (iii) Norm of projection onto the stable subspace.}  When there is no noise, network \ref{first-order} is marginally stable and we only consider the dynamics on the stable subspace of $\R^n$ that is orthogonal to the subspace spanned by $\mathbf{1}_{n}$. For a given $Q \in  \R^{(n-1) \times n}$ whose rows form an orthonormal basis of the disagreement subspace, the norm of the projection of $x(t)$ onto the stable  subspace is a  coherency measure  {\cite{young2010robustness}}
		that is given by
		\begin{equation*}
		\| Q x(t) \|_2^2 = x^T(t) Q^T Q x(t)  
		\end{equation*}
		where $Q \mathbf{1}_{n} = 0$ and $Q^T Q = M_n$. 
	}
	
	It can be proven that 
	our utilized measure of performance (\ref{eq:coherenceDef}) is equal to the square of ${\HH}_2$-norm of the network from $\xi$ to $y$. Thus, we utilize the interpretation of energy of impulse response in order to calculate performance measure of the network \cite{Doyle89}, i.e., we have
	\begin{eqnarray}
	\rhoo(L;\tau) \,=\,\frac{1}{2\pi}\Tr\Big[\int_{-\infty}^{+\infty}{G^{\CT}(j\omega)G(j\omega) \: d\omega}\Big],\label{eq:coherenceDef3}
	\end{eqnarray}
	where $G(s)$ is the transfer function of  (\ref{eq:system}) from $\xi$ to $y$.
	
	\subsection{Problem Statement}
	Our main objective is to explore all possible ways to improve   performance of the time-delay first-order consensus networks \eqref{eq:system} with respect to performance measure \eqref{eq:coherenceDef}. In order to tackle this problem, first we need to quantify performance measure \eqref{eq:coherenceDef}  in terms of the spectrum of the underlying coupling graph of the network. Next, we need to characterize inherent  fundamental limits on the best achievable  performance and classify those networks that can  actually achieve this hard limit. There are only four possible ways to improve performance of network \eqref{eq:system} by manipulating its underlying coupling graph: growing, sparsification, and reweighting. Therefore, we need to investigate under what conditions, performance can be improved in each of these three possible scenarios. Improving performance in presence of time-delay is a challenging task due to the counter-intuitive effects of { connectivity}  on the performance.

	\section{Properties of the Performance in Presence of  Time-Delay}\label{sec-III}\allowdisplaybreaks[3]
	In the following Theorem, we derive an exact expression for the performance measure of the consensus network. 
	
	\vspace{0.2cm} 
	\begin{theorem}\label{coherenceMainTheorem}
		For dynamical network \eqref{eq:system}, the performance measure (\ref{eq:coherenceDef}) can be specified by
		\begin{eqnarray}\label{eq:coherenceValueGen}
		\rhoo(L;\tau)=\frac{1}{2}\Tr\Big[L_o L^{\dagger} \cos(\tau L)\big(M_{\N}-\sin(\tau L)\big)^{\dagger}\Big],
		\end{eqnarray}
		where $L_o = C^{\T}C$. In addition, when the output matrix is equal to the centering matrix, i.e.,  $C=M_n$, the performance measure can be quantified as an additively separable function of Laplacian eigenvalues; in other words, we have the following formula
		\begin{eqnarray}\label{eq:coherenceValue}
		\rhoo(L;\tau)=\sum_{i=2}^{\N} {f_\tau(\lambda_i)},
		\end{eqnarray}
		where 
		\begin{eqnarray}\label{fTau}
		f_\tau(\lambda_i)={\frac{1}{2\lambda_i}\frac{\cos(\lambda_i \tau)}{1-\sin(\lambda_i \tau)}}.
		\end{eqnarray}
	\end{theorem}

	\begin{proof}
		In order to find the performance of network \eqref{eq:system}, we utilize {\BL equation \eqref{eq:coherenceDef3}}
		\begin{eqnarray*}
			\rhoo(L;\tau) &=& \frac{1}{2\pi}\Tr\Big[\int_{-\infty}^{+\infty}{G^{\CT}(j\omega)G(j\omega) \: d\omega}\Big],
		\end{eqnarray*}
		where $G(s)$ transfer function of both (\ref{eq:system}) and (\ref{eq:systemProjected}), i.e.,
		\begin{eqnarray}\label{eq:Gs}
		G(s)~=~ C \Bigg(sI_{\N}+e^{-\tau s}\Big(L+\frac{1}{\tau \N}J_{\N}\Big)\Bigg)^{-1}M_{\N}.
		\end{eqnarray}
		We consider spectral decomposition of Laplacian matrix $L$, which is,
		\begin{eqnarray*}
			L~=~Q \Lambda Q^{\T},
		\end{eqnarray*}
		where $Q=[q_1,q_2, \dots , q_{\N}] \in \mathbb{R}^{n\times n}$ is the orthonormal matrix of eigenvectors and $\Lambda=\diag(\lambda_1,\ldots,\lambda_{\N})$ is the diagonal matrix of eigenvalues. We recall that $\lambda_1=0$ for the reason that the graph is undirected and it has no self-loops.
		Therefore,
		\begin{eqnarray}
		M_{\N}&=&I_{\N}-Q \diag(1,0, \dots,0) Q^{\T}\nonumber\\\label{eq:Mn}
		&=&Q \diag(0,1, \dots,1) Q^{\T},
		\end{eqnarray}
		and
		\begin{eqnarray}
		L&=&Q \diag(0,\lambda_2, \dots, \lambda_{\N}) Q^{\T}. \label{eigen-decom}
		\end{eqnarray}
		Thus,
		\begin{eqnarray}
		L+\frac{1}{\tau \N}J_{\N} &=& Q \diag(\frac{1}{\tau},\lambda_2, \dots, \lambda_{\N}) Q^{\T}\label{eq:LJn},
		\end{eqnarray}
		and substituting (\ref{eq:Mn}) and (\ref{eq:LJn}) into (\ref{eq:Gs}),
		\begin{eqnarray*}
			G(s)~=~C Q \diag\Big(0,\frac{1}{s+\lambda_2 e^{-\tau s}},\dots,\frac{1}{s+\lambda_{\N} e^{-\tau s}}\Big)Q^{\T}.
		\end{eqnarray*}
		Hence, we have
		\begin{align}\label{eq:GHG}
		&\Tr\big[G^{\CT}(j \omega)G(j \omega)\big]\nonumber\\
		=&\Tr\Big[ C^{\T} C Q \diag\Big(0,\frac{1}{-j \omega+\lambda_2 e^{j \tau  \omega}},\dots,\frac{1}{-j \omega +\lambda_{\N} e^{j \tau \omega}}\Big)\nonumber\\
		\:\:&\diag\Big(0,\frac{d \omega}{j \omega+\lambda_2 e^{-j \tau  \omega}},\dots,\frac{1}{j \omega+\lambda_{\N} e^{-j \tau \omega}}\Big)
		Q^{\T}\Big]
		\end{align}
		and by substituting (\ref{eq:GHG}) in the following definition of $\rhoo(L;\tau),$
		\begin{eqnarray}
		& & \hspace{-1.25cm}\rhoo(L;\tau) = \frac{1}{2\pi}\Tr\Big[\int_{-\infty}^{+\infty}{G^{\CT}(j\omega)G(j\omega) \: d\omega}\Big] \nonumber\\
		&& = \frac{1}{2\pi} \sum_{i=2}^{\N}{\int_{-\infty}^{+\infty}\frac{\Delta_{i} \HS d\omega}{\big(j \omega+\lambda_i e^{-j \tau \omega}\big)\big(-j \omega+\lambda_i e^{j \tau \omega}\big)}},\label{eq:H2normCalc}
		\end{eqnarray}
		where $\Delta_i$ is the $i^{\mathrm{th}}$ diagonal element of the matrix  $Q^{\T}L_o Q$. By applying Lemma \ref{H2normOfScalerDelayedSystem} to (\ref{eq:H2normCalc}), we get 
		\begin{equation}
		\rhoo(L;\tau)=\sum_{i=2}^{\N} \frac{\Delta_i}{2\lambda_i}~ \frac{\cos(\lambda_i \tau)}{1-\sin(\lambda_i \tau)}.\label{perf-meas}
		\end{equation}
		From eigenvalue decomposition \eqref{eigen-decom}, definition \eqref{sin-cos}, simultaneous diagonalizability of $L$, $\cos(\tau L)$ and  $\big(M_{\N}-\sin(\tau L)\big)$, and notation \eqref{fTau}, we define the function of Laplacian matrix 
		\begin{eqnarray*}
			f_{\tau}(L) &=& Q \HS \mathrm{diag}(0,f_{\tau}(\lambda_2), \ldots, f_{\tau}(\lambda_n)) \HS Q^T \\
			& = & \frac{1}{2} L^{\dagger} \cos(\tau L)\Big(M_{\N}-\sin(\tau L)\Big)^{\dagger}
		\end{eqnarray*}
		Now, we can  rewrite equality \eqref{perf-meas} in the following compact matrix  operator form:
		\begin{eqnarray*}
			\rhoo(L;\tau) &=& \frac{1}{2}\Tr\Big[L_o L^{\dagger} \cos(\tau L)\Big(M_{\N}-\sin(\tau L)\Big)^{\dagger}\Big].
		\end{eqnarray*}
		In addition, when $C = M_n$, we have
		\begin{eqnarray*}
			\rhoo(L;\tau) &=& \frac{1}{2}\Tr\Big[L^{\dagger} \cos(\tau L)\Big(M_{\N}-\sin(\tau L)\Big)^{\dagger}\Big] 
		\end{eqnarray*}
		and $\Delta_i=1$ for $i \geq 2$, from which one can deduct \eqref{eq:coherenceValue}. 
	\end{proof}
	{
		\begin{remark}
			This result for the case that $C = M_n$ was derived in \cite{hunt2012networkThesis,hunt2012network} using complex analysis techniques. We obtained this independently with a state-space point view for a general output matrix $C$ that is orthogonal to vector of all ones. 
		\end{remark}
		\begin{remark}
			Simultaneous diagonalizability of $\cos(\tau L)$, $\sin(\tau L)$ and $L$ follows from definition of matrix sine and cosine through powers of $(\tau L)$ and the fact that $L$ is diagonalizable. Also, as we showed in the proof of the theorem above, $M_n$ and $L$ share same set of eigenvectors and thus they are simultaneous diagonalizable. Hence, $M_{\N}-\sin(\tau L)$ and $L$ are simultaneous diagonalizable, as well.
		\end{remark}
	}

	When there is no time-delay in the network, i.e., $\tau=0$, our result reduces to those of \cite{young2010robustness,Siami16tac,Patterson:2010}, in which  
	\begin{eqnarray*}
		\rhoo(L;0) =\frac{1}{2} \sum_{i=2}^{\N}  \lambda_i^{-1}.
	\end{eqnarray*}
	
	\vspace{0.2cm}
	
	\begin{theorem}\label{coherenceMainTheoremCor}
		The performance measure \eqref{eq:coherenceDef} for the consensus network (\ref{eq:system}) with a  fixed underlying graph is an increasing function of the time-delay, i.e., the following inequality holds
		\begin{eqnarray*}
			\rhoo(L;\tau_1)~<~\rhoo(L;\tau_2)
		\end{eqnarray*}
		for every $0 \leq \tau_1<\tau_2<\frac{\pi}{2\lambda_{\N}}$. 
	\end{theorem}
	
	\vspace{0.2cm}
	\begin{proof}
		As a means to demonstrate that $\rhoo(L;\tau)$ is increasing in $\tau$, we show that in the stability region, first derivative of performance measure with respect to $\tau$ is positive. Following this  idea, we get
		\begin{eqnarray*}
			& & \hspace{-1.5cm} \frac{d}{d\tau}\rhoo(L;\tau) ~= ~\frac{1}{2}\Tr \Big[- L_o \sin(\tau L) \big(M_{\N}-\sin(\tau L)\big)^{\dagger}\\
			& & \hspace{1.3cm} +  L_o \cos^2(\tau L) \Big(\big(M_{\N}-\sin(\tau L)\big)^{\dagger}\Big)^2\Big],
		\end{eqnarray*}
		for all $\tau \in (0,\frac{\pi}{2 \lambda_{\N}})$.
		Due to simultaneous diagonalizability of $L,  \sin(\tau L), \cos(\tau L)$ and $\big(M_{\N}-\sin(\tau L)\big)$ their product is commutative, and therefore, we have
		\begin{eqnarray*}
			\frac{d}{d\tau}\rhoo(L;\tau) &=& \frac{1}{2}\Tr \big[L_o \big(M_{\N}-\sin(\tau L)\big)^{\dagger}\big]\\
			&=& \frac{1}{2}\Tr \big[C \big(I_{\N}-\sin(\tau L)\big)^{-1}C^{\T}\big],
		\end{eqnarray*}
		where the last equality follows from spectral properties of $L$ and orthogonality of rows of matrix $C$ with respect to $1_{\N}$. Since matrix $\big(I_{\N}-\sin(\tau L)\big)^{-1}$ is positive definite and $C$ has nonzero components, we have
		\begin{equation*}
		\frac{d}{d\tau}\rhoo(L;\tau)>0.
		\end{equation*}
	\end{proof}

	\begin{figure}[t]
		\begin{center}
			\psfrag{x}[c][c]{\small{$\lambda_i$}}        
			\psfrag{y}[c][c]{\small{$f_{\tau}(\lambda_i)$}} 
			\psfrag{z}[c][c]{\small{ }} 
			\psfrag{q}[c][c]{\small{$0$}}        
			\psfrag{w}[c][B]{\small{$\frac{z^*}{\tau}$}} 
			\psfrag{e}[c][B]{\normalsize{$\frac{\pi}{4\tau}$}}        
			\psfrag{r}[c][B]{\small{$\frac{\pi}{2\tau}$}} 
			\psfrag{a}[r][c]{\small{$0$}}        
			\psfrag{s}[r][c]{\small{$5{\tau}$}} 
			\psfrag{d}[r][c]{\normalsize{$10\tau$}}        
			\psfrag{f}[r][c]{\small{$15\tau$}}

			\includegraphics[trim=0 0 0 0.5cm,clip,scale=0.35]{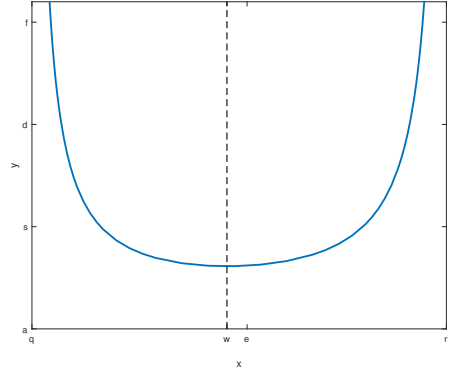}
		\end{center}

		\caption{This plot illustrates convexity property of $\rho_{ss}(L;\tau)$ by depicting $f_\tau(\lambda_i)$ as a function of Laplacian eigenvalues.}
		\label{fig_1}
	\end{figure}

	\begin{theorem}\label{coherenceConvexity}
		For a fixed time-delay, { if 
			$C \in \mathbb{R}^{(n-1)\times n}$ such that rows of $C$ span the disagreement subspace or $C = B_{\mathcal{K}_n}$, i.e. the signed edge-to-vertex incidence matrix of a complete graph}, the performance measure \eqref{eq:coherenceDef} for dynamical network (\ref{eq:system}) is a convex and Schur-convex function of Laplacian eigenvalues of its underlying graph. { In addition, the performance measure is a convex function of weight of links of the underlying graph $\mathcal{G}$.}
	\end{theorem}
	\vspace{0.1cm}
	\begin{proof}
		For all $\lambda \in (0,\frac{\pi}{2\tau})$, the following inequality holds
		\begin{align*}
		\sin(\tau\lambda+\frac{\pi}{4})>\frac{\sqrt{2}}{2}.
		\end{align*}
		By expanding the left hand side of the  inequality above, we get
		\begin{align}\label{eq:380}
		\sin(\tau\lambda)+\cos(\tau\lambda)-1>0.
		\end{align}
		Moreover, by multiplying both sides of \eqref{eq:380} by $\tau\lambda$, we have
		\begin{align*}
		\tau \lambda \sin(\tau\lambda)-\tau\lambda+\tau\lambda \cos(\tau\lambda)>0.
		\end{align*}
		Since $\tau\lambda<2$, it follows that
		\begin{align*}
		\tau \lambda \sin(\tau\lambda)-\tau\lambda+2\cos(\tau\lambda)>0.
		\end{align*}
		Subsequently, we get the following inequality by multiplying both side by $(\sin(\tau\lambda)-1)$ and $\sec^3(\tau\lambda)$, which are respectively negative and positive
		\[
		(\sin(\tau\lambda)-1) \sec^3(\tau\lambda) (\tau\lambda (\sin(\tau\lambda)-1)+2 \cos(\tau\lambda))<0.
		\]
		The left hand side of the above inequality equals to ${\frac{d^2}{d\lambda^2}\Big(\lambda \big(\sec(\tau\lambda)+\tan(\tau\lambda)\big)\Big)}$; thus, we deduce that
		\begin{align}\label{ineqconcave}
		\frac{d^2}{d\lambda^2}\Bigg(\frac{1}{f_\tau(\lambda)}\Bigg)<0.
		\end{align}
		Moreover, from inequality \eqref{ineqconcave} it follows that 
		\begin{align*}
		f_\tau(\lambda)\frac{d^2}{d\lambda^2}f_\tau(\lambda)>0.
		\end{align*}
		Consequently, positiveness of $f_\tau(\lambda)$, results in strict convexity of $f_\tau(\lambda)$. Since the performance function  $\rho_{ss}(L;\tau)$ equals to sum of convex functions; it is a convex function of Laplacian eigenvalues. The Schur-convexity property of $\rho_{ss}(L;\tau)$ follows from its symmetry and Theorem \ref{Schur_H}.\\
		Applying Davis's Theorem \cite{Davis:1957,Borwein:2010}, since the network performance is a symmetric convex function of eigenvalues, it is also a convex function of Laplacian matrix. Therefore, performance measure is a convex function of weight of links of the coupling graph. 
	\end{proof}
	
	Convexity and Schur-convexity properties of the performance function helps us to find fundamental limits on the best achievable performance as well as upper bounds on the performance measure of the network without knowing the spectrum of Laplacian matrix of the underlying graph \cite{siami2014schur}. 
	
	\vspace{0.2cm}

	\section{Optimal and Robust Topologies w.r.t. Time-Delay}\label{sec-IV}

	The following result characterizes the optimal interconnection topology for a consensus network in presence of time-delay.
	\vspace{0.2cm}

	\begin{theorem}\label{hardLimit}
		For the first-order linear consensus network (\ref{eq:system}) with $\N$ nodes that is affected by time-delay $\tau$, the limit on the best achievable performance is given by
		\begin{eqnarray}\label{eq:funLimit1}
		{ \rhoo(L;\tau) \geq \frac{\tau \|C\|_F^2 }{2\big(1-\sin(z^*)\big)},}
		\end{eqnarray}
		where $z^*>0$ is the unique positive solution of
		$\cos(z)=z$. Furthermore, the optimal topology in terms of the performance measure is a complete graph with identical weight 
		\begin{equation}\label{optimalWeight}
		w^*(e)  = \frac{z^*}{n \tau}
		\end{equation}
		for every coupling link $e$. 
	\end{theorem}
	\vspace{0.2cm}
	\begin{proof}
		For a fixed time-delay, performance measure will be minimized if $f_{\tau} (\lambda_i)$ is minimized for all $\lambda_i$, where $i \in \big\{2, \dots, \N \big\}$. By strict convexity and twice differentiability of $f_{\tau} (\lambda_i)$, minimum of performance measure is attained if and only if
		\begin{eqnarray*}
			\lambda_i &=&\argmin_{~~~~~\lambda_i}f_\tau(\lambda_i)\nonumber\\
			&=&\left\{\lambda_i ~\Big|~ \frac{d}{d{\lambda_i}}f_\tau(\lambda_i)=0 \right\} \text{ ~for~ } i \in \big\{2, \dots, \N \big\}.
		\end{eqnarray*}
		Solving $\frac{d}{d{\lambda_i}}f_\tau(\lambda_i)=0$ results in the following equation
		\[ \cos(\lambda_i\tau)=\lambda_i\tau.\]
		Since $z^*$ is solution of the equation $\cos(z)=z$, we have
		\begin{equation}\label{eq:minfTau3}
		\lambda_i=\frac{z^*}{\tau} 
		\end{equation}
		for all $i \in \big\{2, \dots, \N \big\}$. In addition, substituting $\lambda_i$ from the previous equation in \eqref{perf-meas}, we obtain 
		\begin{align*}
		\rhoo(L;\tau) \geq \frac{\tau \|C\|_F^2 }{2\big(1-\sin(z^*)\big)}\sum_{i=2}^{\N}\Delta_i.
		\end{align*}
		Consequently, since $\sum_{i=2}^{\N}\Delta_i = \Tr[Q^{\T}C^{\T}C Q]$,
		the fundamental limit given by inequality  (\ref{eq:funLimit1}) can be deduced. Moreover, considering equality of all eigenvalues of $L$, the underlying graph of the network with the optimal performance is a complete graph with equal link weights.
		Besides, since non-zero eigenvalues of a complete graph with uniform link weights $w^*(e)$, for all links $e$, are
		\begin{equation}\label{eq:minfTau4}
		\lambda_i = \N w^*(e)  
		\end{equation}
		for all $i \in \big\{2, \dots, \N \big\}$, substituting $\lambda_i$ from equation (\ref{eq:minfTau4}) into the equation (\ref{eq:minfTau3}) yields identity  \eqref{optimalWeight}. 
	\end{proof}
	
	{ The lower bound of the coherency for the case that $C= M_{\N}$ was found using numerical analysis in \cite{hunt2012networkThesis,hod2010analytic}. We found the fundamental limit for general output matrix $C$ and studied uniqueness of the limit using convex analysis.}
	
	\vspace{0.2cm}
	When $\tau=0$, it is known that the best achievable performance for linear consensus networks with weighted underlying graphs can be made arbitrarily small \cite{Siami16tac}. This is consistent with the result of Theorem \ref{hardLimit}, since the best achievable performance over all possible  network topologies approaches zero as time-delay goes to zero.
	It is noteworthy that for the first-order linear consensus networks (\ref{eq:system}), the best attainable performance grows linearly with time-delay. Also, under assumption of fixed delay, best achievable performance increases linearly with network size, i.e., its in the order of $\mathcal{O}(n)$. Furthermore, weight of the links in the network with optimal performance is inversely proportional to network size. In the following theorems, we classify graph topologies of robust consensus networks with respect to  time-delay increments. { We also note that if rows of the output matrix $C$ span the disagreement subspace, then the optimal topology would be unique.} 
	\vspace{0.2cm}
	
	\begin{theorem}\label{RobustTop}
		Suppose that ${L}_1$ and ${L}_2$ are Laplacian matrices of coupling graphs of two consensus networks governed by (\ref{eq:system}). If $\lambda_n(L_1)>\lambda_n(L_2)$ and { rows of output matrix $C$ span the disagreement subspace}, then there exists a threshold $\tau^*>0$ such that for all $\tau>\tau^*$ the following ordering holds
		\[
		\rhoo(L_2;\tau) < \rhoo(L_1;\tau).
		\]
		Moreover, the value of $\tau^*$ depends on $L_1$ and $L_2$ and the ouput matrix.
	\end{theorem}
	
	\begin{proof}
		As $\tau$ approaches $\frac{\pi}{2\lambda_n(L_1)}$  from left, $\rhoo(L_1;\tau)$ increases unboundedly. A proof follows from boundedness of $\rhoo(L_2;\tau)$ for $\tau \in \big[0 ,\frac{\pi}{2\lambda_n(L_1)}\big)$ and unboundedness of $\frac{\pi}{2\lambda_n(L_1)}$ in the same interval. 
		The smallest $\tau^*$ is solution of a nonlinear equation, although in the following we show that any 
		\begin{eqnarray}\label{condPerf}
		\tau \in \left[\frac{\pi \hat{p}}{2 \hat{p} \lambda_n(L_1)+1},\frac{\pi}{2\lambda_n(L_1)}\right)
		\end{eqnarray}
		can serve as a $\tau^*$, where $\hat{p} = \rhoo(L_2;\frac{\pi}{2\lambda_n(L_1)})$.
		Based on definition of $\rhoo$, we have
		\[
		\rhoo(L_1;\tau) \geq \frac{\Delta_n}{2\lambda_n(L_1)}\frac{\cos\big(\tau \lambda_n(L_1)\big)}{1-\sin\big(\tau \lambda_n(L_1)\big)}
		\]
		for all $\tau \in \big[0 ,\frac{\pi}{2\lambda_n(L_1)}\big)$.
		Furthermore, since $\frac{1}{2\lambda_n(L_1)}>\frac{\tau}{\pi}$, we have
		\[
		\rhoo(L_1;\tau) \geq \frac{\Delta_n \tau}{\pi}\frac{\pi/2}{\pi/2-\tau\lambda_n(L_1)}.
		\]
		Using inequality above, if condition \eqref{condPerf} hold, then,
		\[
		\rhoo(L_1;\tau) \geq \rhoo\Big(L_2;\frac{\pi}{2\lambda_n(L_1)}\Big).
		\]
		Thus, our desired result follows from the result of Theorem \ref{coherenceMainTheoremCor}.
	\end{proof}
	
	\vspace{0.2cm}
	\begin{theorem}\label{RobustTopPath}
		Suppose that two linear consensus networks with dynamics  (\ref{eq:system}) and unweighted underlying graphs are given: the graph topology of one of them is path, denoted by $\mathcal{P}$,  and the other one has an arbitrary non-path graph topology, shown by $\mathcal{G}$.  { If rows of output matrix $C$ span the disagreement subspace}, then there exists a $\tau^*_\mathcal{G} >0$ such that for all $\tau>\tau^*_\mathcal{G}$ the following ordering holds
		\begin{eqnarray*}
			\rhoo(L_\mathcal{G};\tau)>\rhoo(L_\mathcal{P};\tau).
		\end{eqnarray*}
	\end{theorem}
	\begin{proof}
		For a path graph on $\N$ nodes, we have 
		\begin{eqnarray}\label{lambdaNpath}
		\lambda_n(L_\mathcal{P}) =  2- 2\cos\left(\frac{n-1}{n}\pi\right) < 4.
		\end{eqnarray}
		Whereas, for any  non-path topology $\mathcal{G}$ and  $n>3$, we have
		$$d_{\mbox{max}}(\mathcal{G}) \geq 3.$$
		Thus, for any non-path $\mathcal{G}$ with more than 3 nodes, by Lemma \ref{specRad} and \eqref{lambdaNpath} we get
		\begin{eqnarray*}
			\lambda_n(L_\mathcal{P})< 4 <\lambda_n(L_\mathcal{G}) .
		\end{eqnarray*}
		For graphs with more than 3 nodes a proof follows from Theorem \ref{RobustTop}.
		We note that for $n=3$ there exists only two topologies, namely path graph and complete graph. Even though, the largest eigenvalue of both of these graphs are equal to 3, the multiplicity of this eigenvalue for the complete graph is 2, and the required result follows.  
	\end{proof}
	
	\vspace{0.2cm}
	\begin{theorem}\label{RobustTopRing}
		Suppose that two linear consensus networks with dynamics  (\ref{eq:system}) and unweighted underlying graphs are given: one with ring topology $\mathcal{R}$ and the other one with an arbitrary non-ring topology $\mathcal{G}$ that has at least one loop. If  $C=M_n$, then there exists a $\tau^*_\mathcal{G} > 0$ such that for all $\tau>\tau^*_\mathcal{G}$ the following inequality holds
		\begin{eqnarray*}
			\rhoo(L_\mathcal{G};\tau)>\rhoo(L_\mathcal{R};\tau).
		\end{eqnarray*}
	\end{theorem}

	\begin{proof}
		For ring graphs with $n > 4$ nodes, we have
		
		\[ \lambda_n(L_\mathcal{R}) \leq  2- 2\cos(\pi) = 4.\]
		
		Whereas, for other non-ring, non-tree graph topologies $\mathcal{G}$ with more than 4 nodes, Lemma \ref{specRad} yields
		
		\[\lambda_n(L_\mathcal{R})\leq 4 <\lambda_n(L_\mathcal{G}).\]
		
		For $n>4$, the proof follows from Theorem \ref{RobustTop}. For $n= 4$, our claim holds true since for all non-tree topologies, we have $\lambda_n =4$, but the second largest eigenvalue for a ring graph is $2$ and for the rest of non-tree topologies, it is greater than or equal to $3$. From this the desired result follows.
	\end{proof}
	
	
	The results of Theorems \ref{RobustTop} and \ref{RobustTopRing} are in contrast with the common intuition of performance in non-delayed first-order linear consensus networks, where enhancing connectivity in the sense of $L_1 \preceq L_2$ always results in performance improvement, i.e., $\rhoo(L_2;0) \leq \rhoo(L_1;0)$  and path graph has the worst performance among all unweighted graphs. In conclusion, in  time-delay linear consensus networks, higher connectivity does not necessarily imply better $\HH_2$-norm performance \cite{ghaedsharaf2016interplay}. 
	
	{\BR In the subsequent sections, we exploit properties of the performance measure, such as convexity, to formulate algorithms to enhance the performance. Furthermore, we discuss efficacy of these algorithms by comparing them with fundamental limits that we derived in this section.}

	\section{Improving Performance by Adding New Feedback Loops}\label{sec-V}
	
	In the following three sections, we consider the problem of performance improvement in a time-delay linear consensus network. There are only four possible ways to achieve this objective via manipulating the underlying graph of the network: (i) adding new interconnection links or growing, (ii) adjusting weight of existing links, and (iii) eliminating existing links or sparsification. Other design objectives, such as rewiring, can be equivalently executed in several consecutive design steps involving reweighing, growing, and sparsification. Therefore, we focus our attention on these three core design schemes.
	
	In this section, we consider the problem of growing consensus network (\ref{eq:system}), where it is allowed to establish new interconnections links in the network. It is assumed that some of the Laplacian eigenvalues are located on the left side of the dashed line in Figure \ref{fig_1}, i.e., $\lambda_i < \frac{z^*}{\tau}$ for some $i \in \big\{2, \ldots,n\big\}$. In this case, enhancing the connectivity can improve $\HH_2$-norm performance  of the network. 
	
	Suppose that a set of candidate links $\EE_c$ and a corresponding weight function $\Wc:\EE_c \rightarrow \mathbb{R}_{+}$ are given. Adding a new link between two agents is equivalent to closing a new feedback loop around these two agents according to our earlier interpretation \eqref{feedback-law}. Therefore,  weight of a candidate link plays role of a feedback gain in the overall closed-loop system and it cannot be chosen arbitrary; its value is opted by considering all existing constraints. Based on this elucidation, it is reasonable to consider the following modified form of network (\ref{eq:system}) for our design purpose
	\begin{eqnarray*}
		\label{eq:systemFB} 
		\dot{x}(t) & = & -Lx(t-\tau)+u(t)+\xi(t)\\
		u(t) & =& -\LK x(t-\tau)\\
		y(t) & =& { C}~x(t)
	\end{eqnarray*}
	that can be rewritten in the following closed-loop form 
	\begin{subequations}\label{eq:systemFBC}
		\begin{eqnarray}
			\dot{x}(t) & =& -(L+\LK)x(t-\tau)+\xi(t)\\
			y(t) & =&  { C}~x(t)
		\end{eqnarray}
	\end{subequations}
	where $\LK$ is the Laplacian matrix of the feedback gain and can be represented by  
	\begin{align*}
	\LK = \sum_{e \in \EE_s} \Wc(e) b_e b_e^{\T},
	\end{align*}
	in which  $b_e$ is the corresponding column to edge $e$ in the node-to-edge incidence matrix of the underlying graph of the network. Our design objective is to improve performance of the noisy network in presence of time-delay by designing a sparse Laplacian feedback gain $\LK$ with at most $k$ links  having predetermined weight, i.e., our goal is to solve the following optimization problem 
	
	\vspace{.25cm}
	\hspace{-.5cm}\begin{tabularx}{0.504\textwidth}{l X r}
		$\underset{\EE_s}{\mbox{Minimize}}$   & $\rhoo(L+\LK;\tau)$ & \tagarray \label{optProb1}\\
		\mbox{subject to:} & $\displaystyle\LK = \sum_{e \in \EE_s} \Wc(e) b_e b_e^{\T}$, & \tagarray \label{cond1}\\
		& $0 \preceq L+\LK \prec \frac{\pi}{2\tau} I_{\N}$, & \tagarray \label{cond2}\\
		& $|\EE_s|\leq k \hspace{0.5cm} \textrm{for all} \hspace{0.5cm}\EE_s {
			\subseteq} \EE_c$, & \tagarray \label{cond3}
	\end{tabularx}
	
	\vspace{0.2cm}
	\noindent Condition (\ref{cond2}) ensures stability of the closed-loop network \eqref{eq:systemFBC}. Since the problem given by (\ref{optProb1})-(\ref{cond3}) is combinatorial, the exact solution must be found by an exhaustive search and appraising $\rhoo(L+\LK;\tau)$ for all possible $\sum_{i=1}^k \binom{|\EE_c|}{i}$ cases. In real-world problems, when the size of candidate set is prohibitively large, we need efficient methods to tackle the problem. Furthermore, when there is no time-delay, the $\mathcal{H}_2$-norm performance of the network will improve no matter how we choose and add the new candidate links   \cite{siami2016topology}. However, in presence of time-delay, adding new links may deteriorate performance or even destabilize the closed-loop network, which is why growing a time-delayed network  is a more delicate task.   
	
	\subsection{Cost Function Approximation and SDP Relaxation} 
	
	We can derive a convex relaxation of our problem by letting constants $\Wc(e)$ to become decision variables, shown by $\wc(e)$, and replacing the constraint (\ref{cond3}) by 
	\begin{equation}
	\sum_{e\in \EE_c} \wc(e) ~\leq~ W_k \label{relax-1}
	\end{equation}
	where 
	\begin{equation} W_k = \max\limits_{\substack{\mid \EE_s \mid = k \\ {\EE}_s {
				\subseteq} {\EE}_c }} ~\sum_{e\in {\EE}_s} \Wc(e).\label{relax-2}\end{equation}
	
{\BL \noindent Thus, our design optimization problem is to solve

		\vspace{.25cm}
		\hspace{-.5cm}\begin{tabularx}{0.504\textwidth}{l X r}
			$\underset{ \{ \wc(e)\Sp|\Sp \forall e \in \EE_c \}}{\mbox{Minimize}}$   & $\rhoo(L+\LK;\tau)$ & \tagarray \label{optProb1_dup}\\
	~		\mbox{subject to:}& \eqref{cond1}, \eqref{cond2}, \eqref{relax-1}, \eqref{relax-2} &
		\end{tabularx}
}
		
		{{\BL In spite of smoothness of the cost function in (\ref{optProb1_dup}), the structure of the cost function is not appealing since we cannot cast it as an SDP with linear objective function and constraints or solve it using existing and standard solvers or toolboxes.} Moreover, if we want to write a solver for the problem using the conventional methods, e.g., interior-point or subgradient methods, we have to find eigenvalues and eigenvectors of $L+\LK$ for each step of minimizing the performance function; which significantly increases complexity in terms of both time and details of solver. Therefore, we need an alternative way to remove eigen-decomposition from our solution. We overcome this obstacle by introducing a tight approximation of (\ref{eq:coherenceValueGen}) that has a small relative error with respect to our performance measure.
			
			\begin{lemma}\label{lemma-H-2-approx}
				{\BL For a stable linear consensus network \eqref{eq:systemFBC} with output matrix $C$ with property $C\mathbf{1}_n = \mathbf{0}_n$}, the spectral function 
				\begin{small}\begin{eqnarray*}
						\rhoot(L;\tau) =  \frac{1}{2} \Tr\big[L_oL^{\dagger} + \frac{4\tau}{\pi}L_o\big(\frac{\pi}{2}M_{\N}-\tau L \big)^{\dagger}\!+c_1 \tau^2 L_o L+\frac{c_0}{2} \tau L_o \big], 
					\end{eqnarray*}
				\end{small}
				approximates (\ref{eq:coherenceValueGen}) with guaranteed error bound
				\begin{equation*}
				0\leq\frac{\rhoo(L;\tau) -\rhoot(L;\tau)}{\rhoo(L;\tau)}~\leq~ { 2 \times 10^{-4}},
				\end{equation*}
				where $c_0 = 0.18733$ and ${c_1 =-0.01}$ are
				constants to minimize mean the squared error numerically. 
			\end{lemma}
			\begin{proof}
				{  
					We recall that 
					\begin{align*}
					\rhoo(L;\tau) =~ \sum_{i=2}^{\N}\frac{{ \Delta_i}}{2\lambda_i}\frac{\cos(\tau \lambda_i)}{1-\sin(\tau \lambda_i)},
					\end{align*}
					by multiplying the nominator and denominator by $\tau$ we get
					\begin{align}
					\rhoo(L;\tau) =~& \tau \sum_{i=2}^{\N}\frac{{ \Delta_i}}{2\tau \lambda_i}\frac{\cos(\tau \lambda_i)}{1-\sin(\tau \lambda_i)} \nonumber \\
					=~& \tau \sum_{i=2}^{\N}{ \Delta_i}f_{1}(\tau \lambda_i), \label{rhooF1}
					\end{align}
					where $f_{1} (x) = \frac{1}{2x}\frac{\cos(x)}{1-\sin(x)}$ with domain $x  \in (0,\pi/2)$ based on definition of $f_{\tau}$ in (\ref{fTau}).
					As a means to find a proper approximate performance function, we look for an approximation of $f_{1}$ and we denote it by $\tilde{f}$. Since $f_{1}$ has two vertical asymptotes inside its effective domain, we want to have bounded $\|\tilde{f} -f_1\|_{\infty}$ over effective domain of these functions. To that end, we utilize
\begin{align}
\tilde{f}(x) = \frac{1}{2} \Big( \frac{1}{x} + \frac{4}{\pi} \frac{1}{\pi/2 -x} + c_0 +c_1 x\Big ), \label{approx-fcn}
\end{align}
					as approximation of $f_1$ where $c_0 = 0.18733$ and ${c_1 =-0.01}$ are constants to minimize mean squared error numerically.
					We define our performance approximate function by substituting $\tilde{f}$ for $f_1$ in \eqref{rhooF1}. Thus, from relative error of $\tilde{f}$ with respect to $f_1$ given in Figure \ref{approxFig}, it yields the desired approximation bound. Consequently, we can write the approximation function in the following form
					\begin{equation*}
					\begin{aligned}
					&\rhoot(L;\tau) =  \tau \sum_{i=2}^{\N}{ \Delta_i}\tilde{f}(\tau \lambda_i) \label{rhooTildeF}\\
					&=  \tau \sum_{i=2}^{\N} \frac{{ \Delta_i}}{2} \Big( \frac{1}{\tau \lambda_i} + \frac{4}{\pi} \frac{1}{\pi/2 - \tau \lambda_i} + c_0 +c_1 \tau \lambda_i \Big )\\
					&=  \frac{1}{2} \Tr\big[L_oL^{\dagger} + \frac{4\tau}{\pi}L_o\big(\frac{\pi}{2}M_{\N}-\tau L \big)^{\dagger}\!+c_1 \tau^2 L_o L+\frac{c_0}{2} \tau L_o \big].
					\end{aligned}
					\end{equation*}
				}
			\end{proof}
			{
				\begin{remark}
					Performance of first-order consensus network with $C = M_{\N}$ was previously approximated in \cite{hod2010analytic,hunt2011impact} by using heuristic methods. We have utilized a systematic method to provide an approximate function \eqref{approx-fcn} and we refer to \cite{pranic2013recurrence} for more  details. 
					
				\end{remark}}
				
				\begin{figure}[t]
					\centering
					\psfrag{a}[c][c]{{$\lambda_i$}}        
					\psfrag{b}[c][c]{$\frac{f_1(\lambda_i)- \tilde{f}(\lambda_i)}{f_1(\lambda_i)}$}\includegraphics[width=0.4\textwidth,trim={0.6cm 0 0.6cm 0},clip]{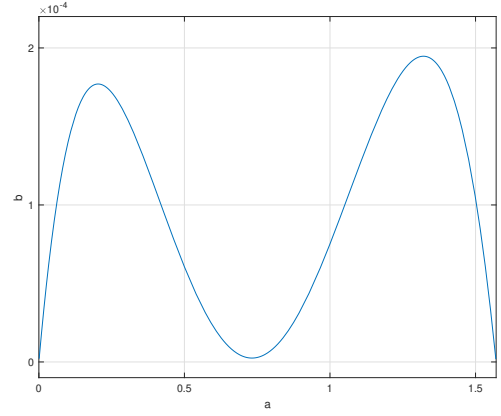}
					\caption{\hspace{0cm}The relative error of $\tilde{f}$ with respect to $f_1$.}
					\label{approxFig}
					\vspace{-.8cm}
				\end{figure}

				{ It is straightforward to show that $\rhoot$ is convex function of eigenvalues and weights of the links for any $C$ such that $C \mathbf{1}_{\N} = \mathbf{0}_{\N}$.} Replacing $\rhoo$ by $\rhoot$ and combinatorial constraint \eqref{cond3} by its relaxed form \eqref{relax-1},  we can relax \eqref{optProb1_dup} to the following optimization problem
				
				\hspace{-.25cm}\begin{tabularx}{0.504\textwidth}{l X r}
					$\underset{ \{\wc(e)\Sp|\Sp \forall e \in \EE_c \}}{\mbox{Minimize}}$   & $\rhoot(L+\LK;\tau)$ & \tagarray \label{optProb1_dup2}\\
~					\mbox{subject to:}& \eqref{cond1}, \eqref{cond2}, \eqref{relax-1}, \eqref{relax-2}.&
				\end{tabularx}}
				
				In addition, neglecting the constant term in $\rhoot$, the optimization problem \eqref{optProb1_dup2} is equivalent to the following SDP
				\vspace{.25cm}
				
				\begin{tabularx}{0.475\textwidth}{l X r}
					\hspace{-.6cm}$\underset{ \{\wc(e)\Sp|\Sp \forall e \in \EE_c \}}{\mbox{Minimize}}$    & $\Tr \big[L_o X_1+\frac{4}{\pi}L_o X_2+c_1 \tau L_o \LK \big]$& \tagarray \label{optProb1SDP}\\
					\hspace{-.7cm} \mbox{subject to:} & $\displaystyle\LK = \sum_{e \in {\EE}_c} \wc(e) b_e b_e^{\T}$ & \tagarray \label{cond1s}\\
					& $\displaystyle\sum_{e\in {\EE}_c} \wc(e) \leq W_k$ & \tagarray \label{cond3s}\\
					& $\begin{bmatrix}
					X_1 & I\\
					I & \tau(L+\LK) +\frac{1}{n}J \end{bmatrix} \succeq 0 $ & \tagarray \label{cond5s}\\
					& $\begin{bmatrix}
					X_2 & I\\
					I & \frac{\pi}{2}I_n-\tau( L +\LK) \end{bmatrix} \succeq 0$. & \tagarray \label{cond6s}\\
				\end{tabularx}

				\vspace{0.2cm}
				
				
				The following Theorem investigates the efficacy of using $\rhoot$ as in optimization problem \eqref{optProb1SDP}-\eqref{cond6s}.
				
				\begin{theorem}\label{optimalityDegree}
					{\BL Let $L_F^*$ be  the  solution of the optimization problem \eqref{optProb1_dup} and $\hat{L}_F$ be the solution of the minimization problem \eqref{optProb1_dup2} where $C = M_n$ or $C$ is an orthonormal matrix $Q$ such that $Q\mathbf{1}_{\N} = \mathbf{0}_{\N}$.} Then, we have
					\begin{eqnarray*}
						\rhoo(L+L_F^*;\tau) \leq \rhoo(L+\hat{L}_F;\tau) \leq  (1+\epsilon) \rhoo(L+L_F^*;\tau),
					\end{eqnarray*}
					where ${ \epsilon = 2 \times 10^{-4}}$.
				\end{theorem}
				
				\vspace{0.1cm}
				\begin{proof}
					First, we look at relative error of $f_1$ with $\tilde{f}$. For ${\epsilon_1 = 2 \times 10^{-4}}$ we have
					\begin{eqnarray*}
						0 \leq\frac{f_1({\lambda}_i)-\tilde{f}({\lambda}_i)}{{f_1}({\lambda}_i)} \leq \epsilon_1.
					\end{eqnarray*}
					In addition, rearranging and doing summation over all $\lambda_i$'s yields
					\begin{eqnarray}\label{53}
					\rhoo({L};\tau) \leq{\rhoot({L};\tau)} \leq (1+\epsilon_1)\rhoo({L};\tau).
					\end{eqnarray}
					Moreover, since $L+\hat{L}_F$ and ${L+L_F^*}$ are minimizers of $\rhoot$ and $\rhoo$, repectively, we observe that
					\begin{eqnarray}\label{54}
					\rhoot(L+\hat{L}_F;\tau) \leq \rhoot({L+L_F^*};\tau),\\\label{55}
					\rhoo({L+L_F^*};\tau) \leq \rhoo(L+\hat{L}_F;\tau).
					\end{eqnarray}
					Inequalities \eqref{53} and \eqref{54} yield that
					\begin{eqnarray}\label{56}
					\rhoo(L+\hat{L}_F;\tau) \leq {(1+\epsilon_1)}\rhoo({L+L_F^*};\tau).
					\end{eqnarray}
					Lastly, from \eqref{55} and \eqref{56} we obtain
					\begin{eqnarray*}
						\rhoo(L+L_F^*;\tau) \leq \rhoo(L+\hat{L}_F;\tau) \leq  (1+\epsilon_1) \rhoo(L+L_F^*;\tau),
					\end{eqnarray*}
					for $\epsilon_1 = 2 \times 10^{-4}$.
				\end{proof}

				\subsection{Greedy Algorithms} 
				In spite of the fact that the SDP relaxation of our problem can be solved using conventional SDP solvers, it cannot be utilized to improve performance of a moderately sized network (more than 20000 candidate edges) as it would require a large amount of memory, which is not practically plausible. To address this issue, and in light of Theorem \ref{optimalityDegree}, we propose  greedy algorithms to tackle the optimal control problem given in \eqref{optProb1}-\eqref{cond3} for moderately sized networks. 
				{
					An undesirable naive procedure for one step of greedy algorithm is to choose the optimal link by evaluating the performance measure after adding the candidate links to the network one at a time, which involves computing the pseudo-inverse of the Laplacian matrix  for each candidate link.} A positive aspect of using $\rhoot$ as performance function is that it spares us the complexity of using eigen-decompsition for  Laplacian matrix. The following Theorem highlights an additional positive aspect of utilizing $\rhoot$ instead of $\rhoo$ that enables us to calculate a useful explicit rank-one update rule.  
				
				\begin{theorem}\label{rankOneUpdate}
					Let $L_e$ be the rank one weighted Laplacian matrix of a graph with only a single edge $e$ between nodes $i$ and $j$ nodes with a given  weight $\Wc(e)$. Then,
					\begin{eqnarray}\label{increment}
					\rhoot(L+L_e;\tau) = \rhoot(L;\tau) + c(e),
					\end{eqnarray}
					where 
					\begin{eqnarray}\label{purturbationVal}
					& & \hspace{-1.8cm} c(e) := -\frac{r_e(LL_o^{\dagger}L)}{2\Wc(e)^{-1}+2 r_e(L)} + c_1 \tau ^2 \Wc(e) \nonumber\\
					& & \hspace{1cm} -\frac{2 \tau}{\pi}\frac{r_e\big((\frac{\pi}{2}M_{\N}-\tau L)L_o^{\dagger}(\frac{\pi}{2}M_{\N}-\tau L)\big)}{-(\Wc(e)\tau)^{-1}+r_e(\frac{\pi}{2}M_{\N}-\tau L)} .
					\end{eqnarray}
				\end{theorem}
				\begin{proof}
					Rearranging \eqref{increment} yields
					\begin{align}\label{rearranged}
					\hspace{-0.4cm} c(e) = \rhoot(L+L_e;\tau) - \rhoot(L;\tau) \nonumber\\
					= \frac{1}{2}\Tr[L_o(L+L_e)^{\dagger}-L_o L^{\dagger}]+\frac{\tau}{2}\Tr[c_1\tau L_o(L+L_e)-c_1 \tau L_o L]  \nonumber\\
					+ \frac{2\tau}{\pi}\Tr[L_o(\frac{\pi}{2}M_{\N}-\tau (L+L_e))^{\dagger} \big]-\Tr[L_o(\frac{\pi}{2}M_{\N}-\tau L)^{\dagger}].
					\end{align}
					Since $L_e$ is a rank-one Laplacian matrix of a graph with a single edge between node $i$ and $j$ with weight $\Wc(e)$, we have
					$${L_e = \Wc(e)(\chi_i-\chi_j)(\chi_i-\chi_j)^{\T}},$$
					and further utilizing the Sherman-Morrison formula \cite{hager1989updating} for rank-one update we have
					\begin{align*}
					\Tr\big[L_o(L+L_e)^{\dagger}&-L_o L^{\dagger}\big]
					=\Tr\Big[L_o L^{\dagger}\\
					&-\frac{\Wc(e) L_o L^{\dagger}(\chi_i-\chi_j)(\chi_i-\chi_j)^{\T}L^{\dagger}}{1+\Wc(e) (\chi_i-\chi_j)^{\T}L^{\dagger}(\chi_i-\chi_j)}-L_o L^{\dagger}\Big]
					\end{align*}
					In addition, using the cyclic permutation property of the trace operator, it yields that
					\begin{align*}
					\Tr\big[L_o(L+L_e)^{\dagger}-L_o L^{\dagger}\big]\!\!\!&=\!\!-\!\frac{\Wc(e) (\chi_i-\chi_j)^{\T}L^{\dagger}L_o L^{\dagger}(\chi_i-\chi_j)}{1+\Wc(e) (\chi_i-\chi_j)^{\T}L^{\dagger}(\chi_i-\chi_j)}\\
					&=-\frac{r_e(LL_o^{\dagger}L)}{\Wc(e)^{-1}+r_e(L)}.
					\end{align*}
					Similarly, applying the Sherman-Morrison formula and the cyclic permutation for the trace to other terms of  \eqref{rearranged}, equation \eqref{purturbationVal} can be obtained.
				\end{proof}

				\vspace{0.1 in}
				
				If we let $\tau = 0$ and { $L_o = M_n$} in \eqref{increment}, we obtain
				\begin{eqnarray*}
					\rhoot(L+L_e;0) = \rhoot(L;0)  - \frac{r_e(L^2)}{2\Wc(e)^{-1}+2 r_e(L)}, 
				\end{eqnarray*}
				which is the contribution of a new edge on the $\mathcal{H}_2$ performance when there is no time-delay \cite{summers2015topology}, \cite{siami2016topology}.
				
				Although $\HH_2$-performance measure of a consensus  network \eqref{eq:system} is not monotone in general with respect to adding new interconnection links to the coupling graph of the network \cite{ghaedsharaf2016interplay}, we can guarantee monotonicity of the $\HH_2$-norm by imposing an upper bound on time-delay. More precisely, let us denote by $\Delta_d$ the maximum possible node degree among all the graphs over the set of all candidate augmented graphs; these are graphs that are obtained by adding $k$ edges from candidate set ${\EE}_c$ to the original graph for all possible choices.
				\begin{lemma}
					In linear consensus  network \eqref{eq:system}, if time-delay satisfies 
					\begin{eqnarray*} 
						\tau \hspace{0.05cm}<~ \frac{z^*}{2 \Delta_d},
					\end{eqnarray*}
					then performance measures $\rho_{ss}$ and $\tilde{\rho}_{ss}$ given in Theorem \ref{coherenceMainTheorem} and Lemma \ref{lemma-H-2-approx}, respectively, will be monotone functions of Laplacian matrix of the network, where $z^*$ is the positive solution of $\cos(z)=z$.
				\end{lemma}
				
				If the performance measure is not monotone, one needs to verify whether adding new interconnection links destabilizes the network.
				
				\begin{theorem}\label{thm-stable}
					Adding a new link $e$ with weight $\Wc(e)$ to network \eqref{eq:systemFBC} will retain stability of the network if and only if 
					\begin{eqnarray}\label{stability new edge}
					\Wc(e) ~<~w_s(L;e), 
					\end{eqnarray}
					where $w_s(L;e)=\big(\tau \: r_e\big(\frac{\pi}{2}M_{\N}-\tau L\big)\big)^{-1}$. 
				\end{theorem}
				
				\begin{proof}
					We first show that if condition \eqref{stability new edge} does not hold, the network becomes unstable. When $\Wc(e)$ approaches $\big(\tau \:\: r_e\big(\frac{\pi}{2}M_{\N}-\tau L\big)\big)^{-1}$ from left, the denominator of the last term in the contribution of new edge to the performance \eqref{purturbationVal} approaches infinity. Consequently, due to boundedness of the approximation function's error with respect to the performance function, as $\wc(e)$ approaches $\big(\tau \: r_e\big(\frac{\pi}{2}M_{\N}-\tau L\big)\big)^{-1}$ from left, the performance function goes to infinity and therefore the network goes to the verge of instability. On the other hand, if condition \eqref{stability new edge} hold, the contribution of a new edge to the performance will be bounded and therefore the performance will stay bounded and thus the system will remain stable.
				\end{proof}
				
				According to Theorem \ref{rankOneUpdate}, the process of calculating the update rule  \eqref{increment} also provides us with the value of quantity $r_e\big(\frac{\pi}{2}M_{\N}-\tau L\big)$ in each step. Therefore, the computational cost of verifying condition \eqref{stability new edge} is negligible.

				In order to set up our Simple Greedy algorithm, we quantify contribution of adding a new edge $e_i \in {\EE}_c \setminus {\EE}_s$ to the performance of the network by 
				\begin{align}
				\hspace{-0.4cm} h_{e_i}({\EE}_s) := &\rhoot \Big(L+ \sum_{e\in {\EE}_s}{\ \Wc}(e) b_{e} b_{e}^{\T} ; \tau \Big)\nonumber\\
				-&\rhoot \Big(L+ \sum_{e\in {\EE}_s\cup \{e_i\}}{ \Wc}(e) b_{e} b_{e}^{\T} ; \tau\Big).\label{eq:h_e}
				\end{align}
				In each step of the algorithm, the edge with maximum contribution to the performance is chosen and added to the coupling graph of the network. All steps of our method are summarized in Algorithm 1.  According to Theorem \ref{thm-stable}, the augmented time-delay linear consensus network from Algorithm 1 is stable and has at most $k$ new links.
				
				\begin{algorithm}[t]
					\caption{Network Growing via Simple Greedy}\label{simpGreedy}
					\begin{algorithmic}[1]
						\BState Initialize: 
						\State \hspace{1.2 cm} ${\EE}_s = \varnothing$
						\State \hspace{1.2 cm} $\LK = 0$ 
						\BState \textbf{for} $i=1\mbox{ to }k$ do:
						\State \hspace{1.2 cm} $e_i = \displaystyle \argmax_{\substack{e \in {\EE}_c \setminus {\EE}_s \\  \Wc(e) < w_s(L+\LK;e)}} h_e({\EE}_s)$
						\vspace{0.2cm}
						\State \hspace{1.2 cm}  \textbf{if} $h_{e_i}({\EE}_s)\leq 0$:
						\State \hspace{2.4 cm}  \textbf{break}
						\State \hspace{1.2 cm}  ${\EE}_s \gets {\EE}_s \cup \{e_i\}$
						\State  \hspace{1.2 cm}  $\LK = \LK + \Wc(e_i)b_{e_i}b_{e_i}^{\T}$
						\BState return ${\EE}_s$
					\end{algorithmic}
				\end{algorithm}

				\begin{remark}
					Except some special cases where value of $h_e({\EE}_s)$ is identical for couple of links, where we should pick them randomly, the rest if the algorithm is deterministic. 
				\end{remark}
				\begin{theorem}
					For linear consensus network \eqref{eq:systemFBC} with a given candidate link { $e \in {\EE}_c \setminus {\EE}_s$}, the performance improvement by adding $e$ to the network { regardless of weight of $e$}, is upper-bounded by
					{ \[ h_e({\EE}_s) ~\leq~ {\frac{r_e(L_s L_o^{\dagger}L_s)}{2 r_e(L_s)}- \frac{c_1\tau}{ r_e\left(\frac{\pi}{2}M_{\N}- \tau L_s\right)}},\]
						where $c_1=-0.01$ is a constant defined in Lemma \ref{lemma-H-2-approx} and ${L_s = L+ \sum_{e\in {\EE}_s\cup \{e_i\}}\Wc(e) b_{e} b_{e}^{\T}}$}. 
				\end{theorem}
				
				\begin{proof}
					From definition of $h_e$ in \eqref{eq:h_e} we have
					\begin{eqnarray}
					& & \hspace{-1.8cm} h_e({\EE}_s)
					~= ~\frac{r_e(L_s L_o^{\dagger}L_s)}{2\Wc(e)^{-1}+2 r_e(L_s)} - c_1 \tau ^2 \Wc(e) \nonumber\\
					& & \hspace{1cm} +\frac{2 \tau}{\pi}\frac{r_e\big((\frac{\pi}{2}M_{\N}-\tau L_s) L_o^{\dagger} (\frac{\pi}{2}M_{\N}-\tau L_s)\big)}{-(\Wc(e)\tau)^{-1}+r_e(\frac{\pi}{2}M_{\N}-\tau L_s)}.\label{Contrib}
					\end{eqnarray}
					In addition, Theorem \ref{thm-stable} states that ${\Wc(e) < \big(\tau \: r_e\big(\frac{\pi}{2}M_{\N}-\tau L_s\big)\big)^{-1}}$, and therefore, by multiplying both sides of the previous inequality by the non-negative constant $-c_1 \tau^2$ we obtain
					\begin{eqnarray}\label{second-term}
					- c_1 \tau ^2 \Wc(e) \leq - c_1 \tau \Big( \: r_e\big(\frac{\pi}{2}M_{\N}-\tau L_s\big)\Big)^{-1}.
					\end{eqnarray}
					Furthermore, stability of the network ensures $\frac{\pi}{2}M_n-\tau L_s \succeq 0$  from which we infer that $\big(\frac{\pi}{2}M_n-\tau L_s\big)^2 \succeq 0$ and consequently, $r_e\big((\frac{\pi}{2}M_{\N}-\tau L_s)^2\big)>0$ for every $e \in {\EE}_c$. Moreover, using Theorem \ref{thm-stable}, we can argue that 
					\begin{eqnarray}\label{third-term}
					\frac{2 \tau}{\pi}\frac{r_e\big((\frac{\pi}{2}M_{\N}-\tau L_s)L_o^{\dagger}(\frac{\pi}{2}M_{\N}-\tau L_s)\big)}{-(\Wc(e)\tau)^{-1}+r_e(\frac{\pi}{2}M_{\N}-\tau L_s)}\leq 0
					\end{eqnarray}
					A proof follows from combining inequalities \eqref{second-term} and \eqref{third-term} with identity \eqref{Contrib}.
				\end{proof}
				
				In literature, other variants of greedy algorithms such as Random Greedy are used to maximize a submodular non-monotone problems. Even though our performance measure is not a supermodular set function, our simulations show that Random Greedy works well and in some cases slightly outperform the Simple Greedy. Furthermore, their consistent outcome can be interpreted as a positive sign that both algorithms work fine. Based on Theorem \ref{thm-stable}, the augmented time-delay linear consensus network from Algorithm 2 is stable and has at most $k$ new links.

				\begin{algorithm}[t]
					\caption{Network Growing via Random Greedy}\label{randGreedy}
					\begin{algorithmic}[1]
						\BState Initialize: 
						\State \hspace{1.2 cm} ${\EE}_s = \varnothing$
						\State \hspace{1.2 cm}  $\LK = 0$
						\State \hspace{1.2 cm} add $2k-1$ dummy members${^\star}$ to ${\EE}_c$
						\BState \textbf{for} $i=1\mbox{ to }k$ do:
						\State \hspace{1.2 cm} $M_i = \displaystyle\argmax_{\substack{M_i \subset {\EE}_c \setminus {\EE}_s , \mid M_i\mid = k \\  \Wc(e) < w_s(L+\LK;e),\HS  \forall e \in M_i}} ~\displaystyle\sum_{e \in M_i} h_e({\EE}_s)$
						\State \hspace{1.2 cm} {\bf repeat} choose $e_i$ randomly  and uniformly from $M_i$
						\State  \hspace{1.2 cm} ${\EE}_s \gets {\EE}_s \cup \{e_i\}$
						\State  \hspace{1.2 cm} $\LK = \LK + \Wc(e_i)b_{e_i}b_{e_i}^{\T}$
						\BState return ${\EE}_s$
					\end{algorithmic}
					\vspace{0.2cm}
					{\small $^\star$These members are edges with zero weight.}
				\end{algorithm}

				\subsection{Time Complexity Analysis}
				We provide a time complexity analysis based on the fastest state-of-the-art algorithms in the literature.
				First, we need to find the pseudo-inverse for $L$ and $\left(\frac{\pi}{2}M_{\N}- \tau L\right)$ which has complexity of $\mathcal{O}(n^3)$ and then, we calculate $L^{\dagger} L_o L^{\dagger}$ and $(\frac{\pi}{2}M_{\N}-\tau L)^{\dagger} L_o (\frac{\pi}{2}M_{\N}-\tau L)^{\dagger}$ which needs $\mathcal{O}(n^3)+\mathcal{O}(n^3+\numRowsOfC n^2)$. For all other steps, using the Sherman-Morrison formula \cite{hager1989updating} for rank-one update, we can find the update for $r_e(L)$, $r_e(L L_o^{\dagger} L)$, $r_e\left(\frac{\pi}{2}M_{\N}-\tau L\right)$ and $r_e\left((\frac{\pi}{2}M_{\N}-\tau L) L_o^{\dagger}(\frac{\pi}{2}M_{\N}-\tau L)\right)$ for all $e \in {\EE}_c$ in $\mathcal{O}(n^2)$. Then, finding the contribution for each link takes constant time for each link. In conclusion, the first step needs $\mathcal{O}(n^3)$ and the rest of the steps take $\mathcal{O}(n^2)$ which is less time comparing to the method of \cite{summers2015topology} which considers network design in absence of time-delay and the essence of their algorithms are similar to ours. The algorithm in \cite{summers2015topology} despite using the Sherman-Morrison formula, eventually needs $\mathcal{O}(n)$ to find the contribution of each link in each step and since we have $\mathcal{O}(n^2)$ candidate links, their algorithm has time complexity of  $\mathcal{O}(n^3)$ in all steps.  The Random Greedy algorithm needs $\mathcal{O}(|{\EE}_c|\log_2k )$ more arithmetic operations than the Simple Greedy in each step, since we have to find the top $k$ contributing links.
				
				\section{Improving Coherency by Adjusting Feedback Gains}\label{sec-VI}
				The second possible way to improve the performance of network \eqref{eq:system} is by adjusting link weights in the underlying graph of the network. This option is domain specific and depends on the underlying dynamics of the network and practical relevance of the problem.    
				In absence of time-delay, the optimal re-weighting problem can be equivalently cast as the effective resistance minimization problem \cite{ghosh2008minimizing}.
				
				As it can be inferred from the  previous section, although the problem of re-weighting the coupling weights is a convex optimization problem, using the  approximate performance measure greatly speeds up the rate of finding the optimal solution. Let denote the total weight of the links in the initial network by $W_{\mathrm{total}}$.  Then, the following semidefinite programming finds the Laplacian matrix for the optimal network in terms of the approximate performance measure:
				
				\vspace{.1cm}
				\hspace{-.5cm}\begin{tabularx}{0.504\textwidth}{l X r}
					\mbox{minimize\hspace{-0.3cm}}   & $\Tr \big[L_o X_1+\frac{4}{\pi}L_o X_2+c_1 \tau L_o \hat{L} \big]$& \\
					\mbox{subject to:\hspace{-0.3cm}} & $\displaystyle \hat{L} = \sum_{e \in {\EE}} \hat{w}(e) b_e b_e^{\T}$ \\
					& $\displaystyle\sum_{e\in {\EE}} \hat{w}(e) \leq W_{\mathrm{total}}$ \\
					&$\hat{w}(e) \geq 0$ \hspace{0.35cm} for all $e \in {\EE}$& \tagarray \label{reweightProb}\\
					& $\begin{bmatrix}
					X_1 & \hspace{-0.2cm}I\\
					I & \hspace{-0.2cm}\tau \hat{L} +\frac{1}{n}J \end{bmatrix}\hspace{-0.1cm} \succeq\hspace{-0.1cm} 0 $, \hspace{-0.2cm} $\begin{bmatrix}
					X_2 & I\\
					I & \frac{\pi}{2}I_n-\tau \hat{L}  \end{bmatrix} \hspace{-0.1cm}\succeq\hspace{-0.1cm} 0$, 
				\end{tabularx}
				where the optimization variables are matrices $X_1, X_2 \in \mathbb{R}^{n \times n}$ and nonnegative continuous variable  $\hat{w}(e)$ as coupling weights for all $e \in {\EE}$. 
				
				\begin{remark}\label{reweigh}
					{\BL For the consensus network \eqref{eq:system},} the resulting network from solving the network reweighing problem \eqref{reweightProb} always have the same or better performance (with respect to $\rhoot)$ than the original network.   
				\end{remark}
				
				It is also possible to improve the  performance measure of the network (\ref{eq:system}) by reweighting the links  while keeping ratio of weight of every two link unchanged.
				The following theorem elaborates more on this approach.
				
				\vspace{0.1cm}
				\begin{theorem}\label{coherenceKOptimization}
					Suppose that for the consensus network (\ref{eq:system}), { rows of the output matrix $C$ span the disagreement space.} Our objective is to improve performance of the network by scaling weights of all links by a constant $\kappa \in (0,\frac{\pi}{2\lambda_{\N}\tau})$. Then, there exists a unique $\kappa^* > 0$ such that 
					\begin{eqnarray*}
						\rhoo(\kappa^*L;\tau) ~ \le ~\rhoo(\kappa L;\tau),
					\end{eqnarray*}
					for all  $\kappa \in (0,\frac{\pi}{2\lambda_{\N}\tau})$. Moreover, the optimal $\kappa^*$ belongs to the interval $[\frac{z^*}{\lambda_n},\frac{z^*}{\lambda_2}]$, where $\lambda_2$ and $\lambda_n$ are the second smallest and the largest eigenvalues of $L$.
				\end{theorem}
				\begin{proof}
					For a fixed underlying graph, let ${h: (0,\frac{\pi}{2\lambda_{\N}\tau}) \to \mathbb{R}}$ be a function where
					\begin{eqnarray*}
						g(\kappa)=\rhoo(\kappa L;\tau).
					\end{eqnarray*}
					From Theorem \ref{coherenceConvexity}, it follows that $g(\kappa)$ is a strictly convex function of $\kappa$. In addition, since $g(\kappa)$ is not a monotonic function of $\kappa$ and is continuous on its domain, it must have an { attainable} minimum. Furthermore, by strict convexity of $g(\kappa)$, a unique positive $\kappa^*$ exists where  $g(\kappa)$ attains its minimum.  
					In addition, if we scale the weights by any $\kappa < \frac{z^*}{\lambda_n}$, all the new eigenvalues will be on the left hand side of the dashed line in Figure \ref{fig_1} and therefore $\kappa =\frac{z^*}{\lambda_n}$ is better than any $\kappa \in (0,\frac{z^*}{\lambda_n})$. Similarly, if we scale the weights by any $\kappa > \frac{z^*}{\lambda_2}$, all the new eigenvalues will be on the right hand side of the dashed line in Figure \ref{fig_1} and thus  $\kappa > \frac{z^*}{\lambda_2}$ will not be better than $\kappa =\frac{z^*}{\lambda_2}$. Finally, by convexity of $g(\kappa)$, we conclude that $\kappa^* \in [\frac{z^*}{\lambda_n},\frac{z^*}{\lambda_2}]$.
				\end{proof}
				
				\vspace{0.15cm}
				Since in the scaling method, we are not using rank-one update of the Laplacian matrix and we need to compute the spectrum of the underlying graph only once, we may use the $\rhoo$ for finding the $\kappa^*$. Thus, this method is the only approach in this paper that needs spectrum of the underlying graph, which can be computed in $\mathcal{O}(n^3)$. Knowing the Laplacian eigenvalues, $\kappa^*$ can be found by exploiting simple techniques such as golden search method, using which,  reaching any $\sigma$-neighborhood of $\kappa^*$ (i.e., finding  $\hat{\kappa}$ such that $|\hat{\kappa} -\kappa^*|<\sigma$) has computational complexity of  order $\mathcal{O}(n \log \frac{1}{\tau \sigma})$.

				\section{Improving Coherency by Feedback Sparsification}\label{sec-VII}
				Suppose that we are required to remove  some of the existing interconnection links in linear consensus network \eqref{eq:system}. Sparsification can potentially happen in practice for several legitimate reasons, including, when there is a budget constraint on communication cost or an enforced security and/or privacy protocol \cite{rezazadeh2018privacy} among the agents that limit each agent to communicate with certain number of neighbors. We assume that some of the Laplacian eigenvalues of the underlying graph are on the right hand side of the dashed line in Figure \ref{fig_1}, i.e., $\lambda_i > \frac{z^*}{\tau}$ for some $i \geq 2$, where under this condition, eliminating some of the existing links can improve performance of the network. {\BL Edge elimination as a method to improve convergence speed of  time-delay consensus networks was previously presented in \cite{koh2016achieving}, where authors measure the convergence rate through simulations in time domain.}
				
				The challenges of this approach are twofold. First, dropping links may break connectivity of the network. By the min-max theorem \cite{teschl2014mathematical}, dropping links from the underlying graph of the network \eqref{eq:system} does not increase Laplacian eigenvalues, and therefore, if largest eigenvalue of the Laplacian is less than $\frac{\pi}{2\tau}$, after dropping links it will remain less than $\frac{\pi}{2\tau}$. Therefore, a failure in connectivity increases number of connected components of the underlying graph. As a result in the disconnected network we will have connected components whose dynamics are decoupled and eigenvalues of each component will remain less than $\frac{\pi}{2\tau}$ as they are a subset of eigenvalues of the whole network. Therefore, each connected component will have its own consensus point. This implies that a failure in connectivity will result in boundlessness of the performance measure. Second, the sparsification problem is inherently combinatorial as we have to find a subset of the current interconnection links in the network and remove them. In the following, we propose remedies to these challenges.  
				
				To tackle the connectivity problem, we must ensure that the link that is being removed is not one of the cut-edges; in other words, removing that specific link would not increase number of connected components of the underlying graph. There exist bridge finding algorithms in an undirected graph, such as Tarjan's Bridge-finding algorithm, which runs in linear time \cite{tarjan}. However, since in each step of our greedy algorithm we have the effective resistance between every two nodes, we can effectively use this existing information to ensure that a selected candidate edge is not a cut-edge.  {\BL In order to avoid removing cut-edges, we must ensure that for a candidate edge $e = \{i,j\}$ the following condition holds
\[\omega(e) \neq \frac{1}{r_e(L)},\] 
where $r_e$ is the effective resistance between node $i$ and $j$\cite{klein1993resistance}}.

				\begin{algorithm}[t]
					\caption{Network Sparsification via Simple Greedy}\label{simpGreedySparse}
					\begin{algorithmic}[1]
						\BState Initialize: 
						\State \hspace{1.2 cm} ${\EE}_s = \varnothing$
						\State \hspace{1.2 cm} $\LK = 0$ 
						\BState \textbf{for} $i=1\mbox{ to }k$ do:
						\State \hspace{1.2 cm} ${\EE}_b = \big\{e ~\big|~ w(e) = r_{e}^{-1}\big(L- \LK \big)\big\}$
						\State \hspace{1.2 cm} $e_i = \displaystyle \argmax_{e \in {\EE} \setminus ({\EE}_s \cup {\EE}_b )} \hat{h}_e({\EE}_s)$
						\State \hspace{1.2 cm}  \textbf{if} $\hat{h}_{e_i}({\EE}_s)\leq 0$:
						\State \hspace{2.4 cm}  \textbf{break}
						\State \hspace{1.2 cm}  ${\EE}_s \gets {\EE}_s \cup \{e_i\}$
						\State  \hspace{1.2 cm} $\LK = \LK - w(e_i)b_{e_i}b_{e_i}^{\T}$
						\BState return ${\EE}_s$
					\end{algorithmic}
				\end{algorithm}
				%
				%

				In order to tackle combinatorial difficulty of sparsification, we utilize the following tailored greedy algorithm to sparsify the coupling graph of a given consensus network. In the algorithm below, ${\EE}_s$ is the set of the coupling links that are chosen to be removed. For an edge $e_i \in {\EE}_c \setminus {\EE}_s$, we denote contribution of removing that edge to performance of the network by the following quantity 
				\begin{eqnarray*}
					\hat{h}_{e_i}({\EE}_s) := \rhoot \Big(L- \hspace{-2mm}\sum_{e\in {\EE}_s}w(e) b_{e} b_{e}^{\T} ; \tau \Big)-\rhoot \Big(L- \hspace{-4.5mm}\sum_{e\in {\EE}_s\cup \{e_i\}}\hspace{-2mm}w(e) b_{e} b_{e}^{\T} ; \tau\Big),
				\end{eqnarray*}
				and in each step of our greedy, we choose an edge that is not a cut-edge and its removal has maximum contribution to the performance. Algorithm 3 summarizes all steps of our greedy method.

				We use $\rhoot$ in Algorithm 3 in order to take advantage of Sherman-Morrison formula and avoid costly eigen-decomposition. Therefore, the process of adding the first link requires computation of  $r_e(L)$, $r_e(L L_o^{\dagger} L)$, $r_e\left(\frac{\pi}{2}M_{\N}-\tau L\right)$ and $r_e\left((\frac{\pi}{2}M_{\N}-\tau L)L_o^{\dagger}(\frac{\pi}{2}M_{\N}-\tau L)\right)$, which can be accomplished with $\mathcal{O}(n^3+\numRowsOfC n^2)$ arithmetic operations. { Knowing the effective resistance between all nodes, by utilizing the aforementioned theorem, it takes $\mathcal{O}(1)$ operations to ensure that an edge is not a cut-edge. As a result, cut-edge verification takes $\mathcal{O}(n^2)$ operation.} For every other step of Algorithm 3, our method needs $\mathcal{O}(n^2)$ operations to  update effective resistance matrices and ensure connectivity by not eliminating a cut-edge.
				
				\begin{remark}
					Our notion of sparsification is basically different from spectral sparsification of  \cite{spielman2011graph}. In our approach, we only eliminate some of the coupling links without re-weighting the remaining links. In addition, our goal is not to create a sparse network with similar performance, but it is to achieve better performance by reshaping the spectrum of the underlying graph of the network.  
				\end{remark}

				\begin{figure}[t]
					\subfloat[\label{G1}]{
						\includegraphics[trim=-1cm 0 -2.5cm 1cm,width=0.13\textwidth]{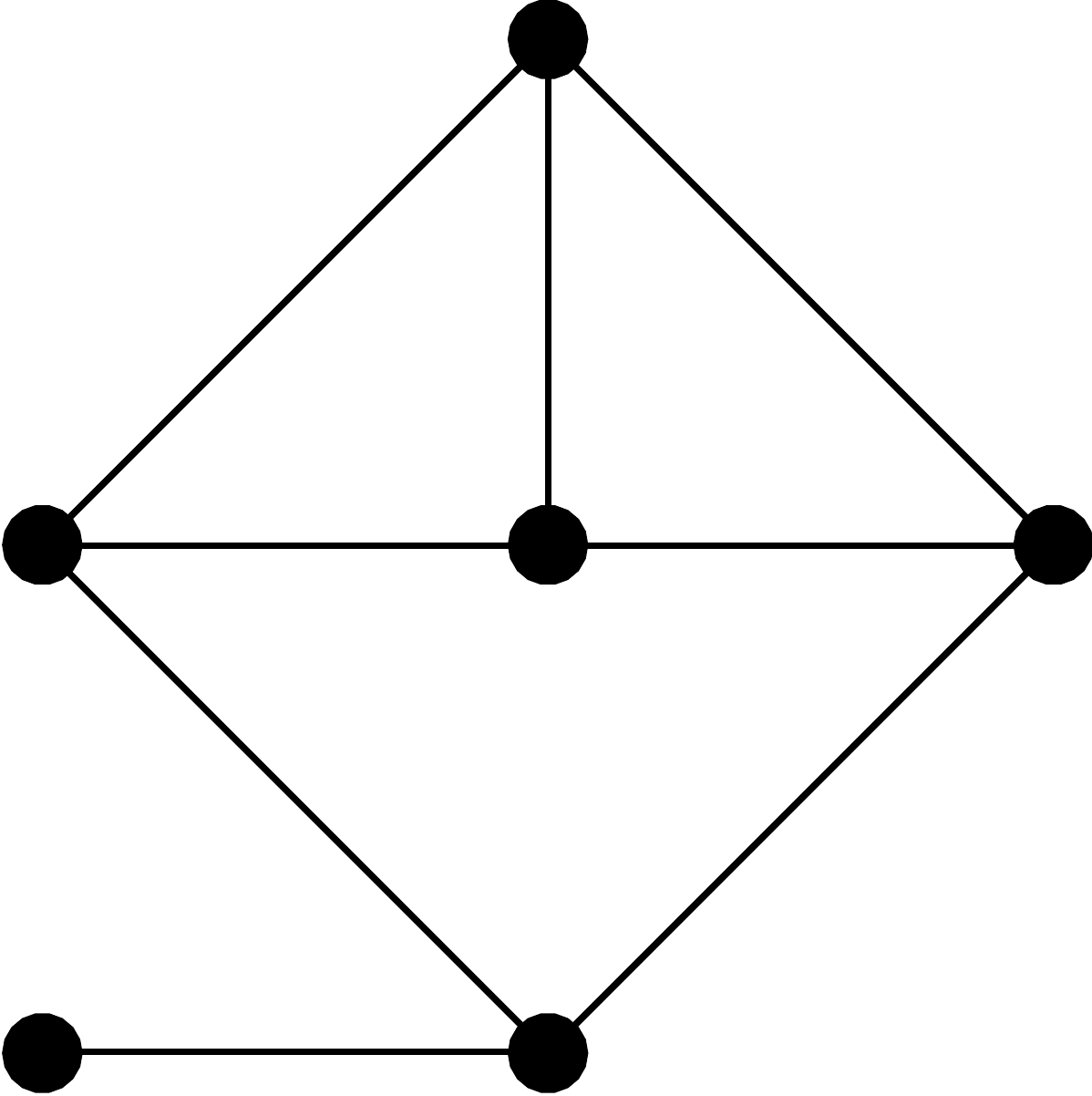}
					}
					\subfloat[\label{G2}]{
						\includegraphics[trim=-1cm 0 -2.5cm 1cm,width=0.13\textwidth]{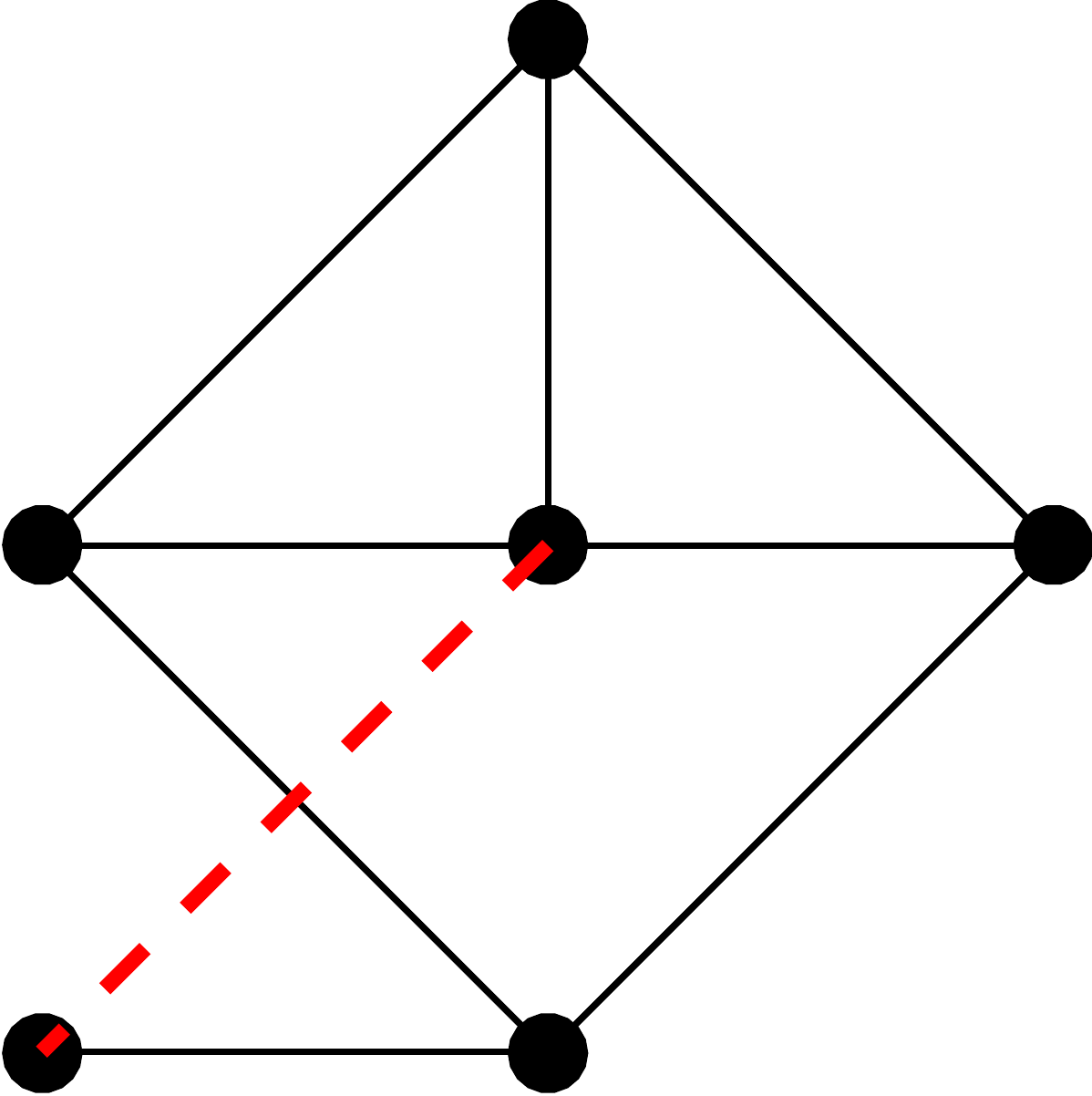}
					}
					\subfloat[\label{G3}]{
						\includegraphics[trim=-1cm 0 -2.5cm 1cm,width=0.13\textwidth]{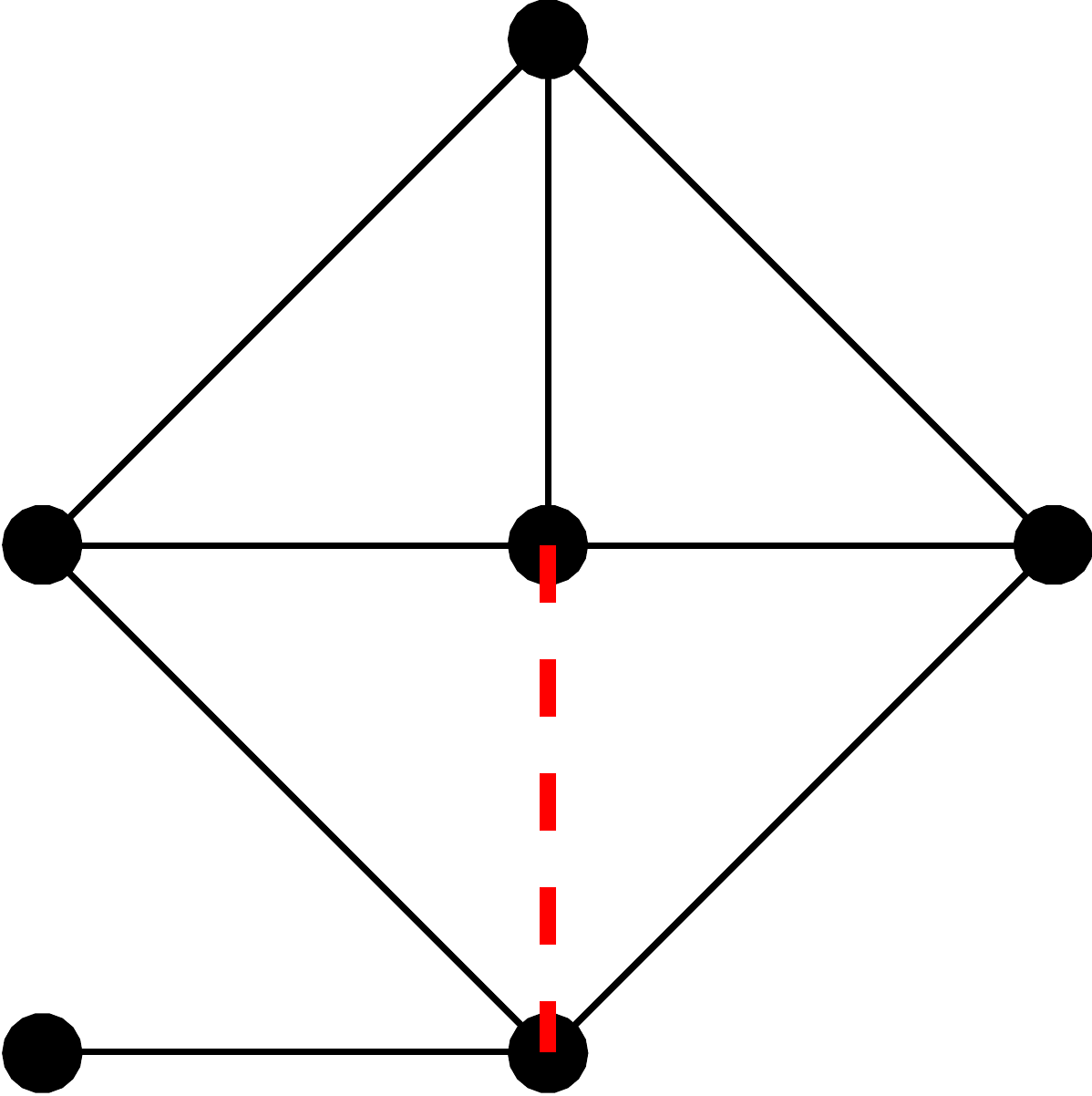}
					}
					\caption{(a) Graph 1, (b) Graph 2, and (c) Graph 3 are unweighted graphs. (b) and (c) are constructed by adding a link to (a).}
					\label{fig_3}
				\end{figure}

				\section{Sensitivity Analysis of the Feedback Structure}\label{sec-IX}
				As we discussed earlier, there are three ways to improve performance. For a given network, our remaining task is to determine which one of the proposed methods should be employed to improve performance of the network. Our reweighing procedure in Section \ref{sec-VI} never deteriorates network's performance, i.e., either it improves it or does not change the weights.
				\begin{theorem}
					Suppose that $e$ is the edge between node $i$ and mode $j$ in the coupling graph of network (\ref{eq:system}). Then, sensitivity of $\rhoot$ with respect to weight of link $e$ is equal to 
					\begin{align}
					\frac{d\rhoot(L;\tau)}{d{\wc}(e)} = &\frac{c_1\tau^2}{2}r_e(L_o^{\dagger}) -\frac{1}{2}r_e(L L_o^{\dagger} L)\nonumber\\
					&+\frac{2\tau^2}{\pi}r_e\big((\frac{\pi}{2}M_{\N}-\tau L) L_o^{\dagger}(\frac{\pi}{2}M_{\N}-\tau L)\big),\label{sensing}
					\end{align}
					where $c_1=-0.01$ is a constant defined in Lemma \ref{lemma-H-2-approx}. 
				\end{theorem}
				\begin{proof}
					Taking partial derivative of $\rhoot$ with respect to weight of edge $e$, we have 
					\begin{align}
					\frac{d\rhoot(L;\tau)}{d\wc(e)} = \frac{{d} }{{d}\wc(e)} \Big(
					\frac{1}{2} \Tr\big[L_o L^{\dagger} + \frac{4\tau}{\pi} L_o \big(\frac{\pi}{2}M_{\N}-\tau L \big)^{\dagger}\!\nonumber&\\
					+c_1 \tau^2 L_o L\big]\Big).&
					\label{sense}
					\end{align}
					Then, we find derivative of all terms in the right hand side of the equation above as follows
					\begin{align*}
					&\frac{d \Tr[L_o L^{\dagger}]}{d\wc(e)} = -\Tr[L_o L^{\dagger}(\chi_i-\chi_j)(\chi_i-\chi_j)^{\T}L^{\dagger}]=- r_e(L L_o^{\dagger} L),\nonumber\\
					&\frac{d \Tr[L_o (\frac{\pi}{2}M_{\N}-\tau L)^{\dagger}]}{d\wc(e)} = \tau \Tr[L_o (\frac{\pi}{2}M_{\N}-\tau L)^{\dagger}(\chi_i-\chi_j)\nonumber\\
					&(\chi_i-\chi_j)^{\T}(\frac{\pi}{2}M_{\N}-\tau L)^{\dagger}]
					= \tau r_e\big((\frac{\pi}{2}M_{\N}-\tau L)L_o^{\dagger}(\frac{\pi}{2}M_{\N}-\tau L)\big).\nonumber
					\end{align*}
					Substituting identities above in  \eqref{sense} we obtain
					\begin{align*}
					\frac{d\rhoot(L;\tau)}{d{\wc}(e)} = &\frac{c_1\tau^2}{2}r_e(L_o^{\dagger}) -\frac{1}{2}r_e(L L_o^{\dagger} L)\nonumber\\
					&+\frac{2\tau^2}{\pi}r_e\big((\frac{\pi}{2}M_{\N}-\tau L) L_o^{\dagger}(\frac{\pi}{2}M_{\N}-\tau L)\big).
					\end{align*}
				\end{proof}
				
				\begin{figure}[t]
					\centering{
						\psfrag{a}[c][c]{ \footnotesize{$\tau$}}        
						\psfrag{b}[c][c]{$\rhoo(L;\tau)$}
						\psfrag{c}[c][c]{ }
						\includegraphics[scale=0.42,trim={0 0 0 -30pt}]{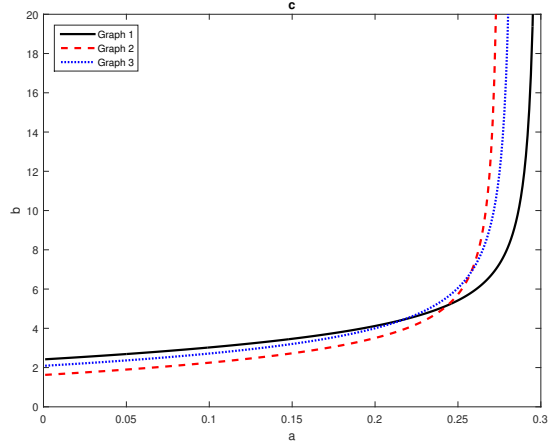}}
					\caption{Comparing performance  of three different topologies shown in Figure \ref{fig_3} as a function of time-delay.}
					\label{fig_4}
				\end{figure}

				Equality \eqref{sensing}
				turns out to be useful in distinguishing whether  sparsification or growing the underlying graph can be effective methods to improve the performance measure. Although, it is also possible to figure out whether removing or adding interconnection can improve the performance by finding the spectrum of the Laplacian matrix in each step, \eqref{sensing} is important since it does not need eigendecomposition and can be updated in each iteration in $\mathcal{O}(n^2)$. Therefore, if adding and removing connections was an option, adding interconnections will improve the performance only if $\frac{d\rhoot(L;\tau)}{d\wc(e)}<0$ for some edge $e \in {\EE}_c$. Similarly, sparsifying the network can be useful only if $\frac{d\rhoot(L;\tau)}{d\wc(e)}>0$ for some edge $e \in {\EE}$. Besides, we can use this approach as a heuristic for the rewiring problem. 
				
				When we are given a set of weightless candidate links, the contribution of each link to the performance cannot be evaluated. If size of the problem (i.e. number of nodes and size of candidate set) is small and we are not concerned  about sparsity of the solution, we can use SDP given by \eqref{optProb1SDP}-\eqref{cond6s} to find  optimal link weights. However, if sparsity is an issue and our objective is  to add at most $k$ new links, we can use identity \eqref{sensing} for adding new links. For the first link, the procedure includes finding $e_1 = \argmin_{e \in {\EE}_c}{\frac{d\rhoot(L;\tau)}{d\wc(e)}}$; this is the link for which the performance measure has the most sensitivity with respect to its weight. As it was discussed earlier, we must have  $\frac{d\rhoot(L;\tau)}{d\wc(e)} <0$. Otherwise, the procedure will be terminated as adding new edges cannot improve $\rhoot$. In addition, we can find the best weight for the selected edge $e_1$ by minimizing \eqref{purturbationVal} over $\Wc(e_1)$ subject to $0\leq\Wc(e_1)<w_s(L;e_1)$,  which can be done in constant time.  Then, we initialize $\EE_s = \{e_1\}$. In order to identify the $i^{\mathrm{th}}$ link, we set 
				\[e_i = \argmin_{e \in {\EE}_c\setminus{\EE}_s}{\frac{d\rhoot\big(L+\sum_{e \in \EE_s} \Wc(e) b_e b_e^{\T};\tau\big)}{d\wc(e)}}\] 
				and maximize $h_{e_i}(\EE_s)$ over $\Wc(e_i)$. Subsequently, we add $e_i$ to $\EE_s$ and procedure continues until we have added $k$ edges or 
				\[\frac{d\rhoot\big(L+\sum_{e \in \EE_s} \Wc(e) b_e b_e^{\T};\tau\big)}{d\wc(e)} \geq 0\] 
				for all $e \in \EE_c\setminus\EE_s$.

				\section{Numerical Examples}\label{sec-VIII}

				In this section, we consider the following numerical examples to demonstrate utility and veracity of our theoretical results, where the data for graph Laplacians of all examples can be downloaded from the following link:
				\begin{center}
					\url{http://www.lehigh.edu/~yag313/TimeDelayGraphs.zip}
				\end{center}
				\vspace{0.1cm}
				
				\begin{example}
					As a means to compare performance of a network in presence and absence of time-delay, we show that in presence of delay, adding a link in two different locations in the network has contrasting effect on performance of the network. Unweighted graphs in Figure  \ref{G3} and Figure  \ref{G2} are constructed by adding one link to distinct locations of the graph  shown in Figure  \ref{G1}. When $\tau=0.235$, performance measure of network with underlying graph of Figure  \ref{G2} is better than the original network and both are better than the network with underlying graph of Figure  \ref{G3}. Nevertheless, without the delay, networks with underlying graph Figure  \ref{G2} and Figure  \ref{G3} both perform better than the original network in terms of noise propagation quantified by $\mathcal{H}_2$-norm. In order to further clarify effect of connectivity in presence or absence of time-delay, in Figure  \ref{fig_4} we drew performance of the three aforementioned network as a function of time-delay. It is noteworthy that when $\tau<0.2$ consensus network with the coupling graph given in Figure  \ref{G3} has a better performance than a network with coupling graph given in Figure  \ref{G1}. Whereas, as the time-delay increases, network with underlying graph in Figure  \ref{G1} starts to outperform network with graph given in \ref{G3}. 
				\end{example}

				\begin{figure}[t]
					\centering
					\includegraphics[width=0.45\textwidth,trim={0.5cm 1cm 1cm 1.6cm}, clip=true]{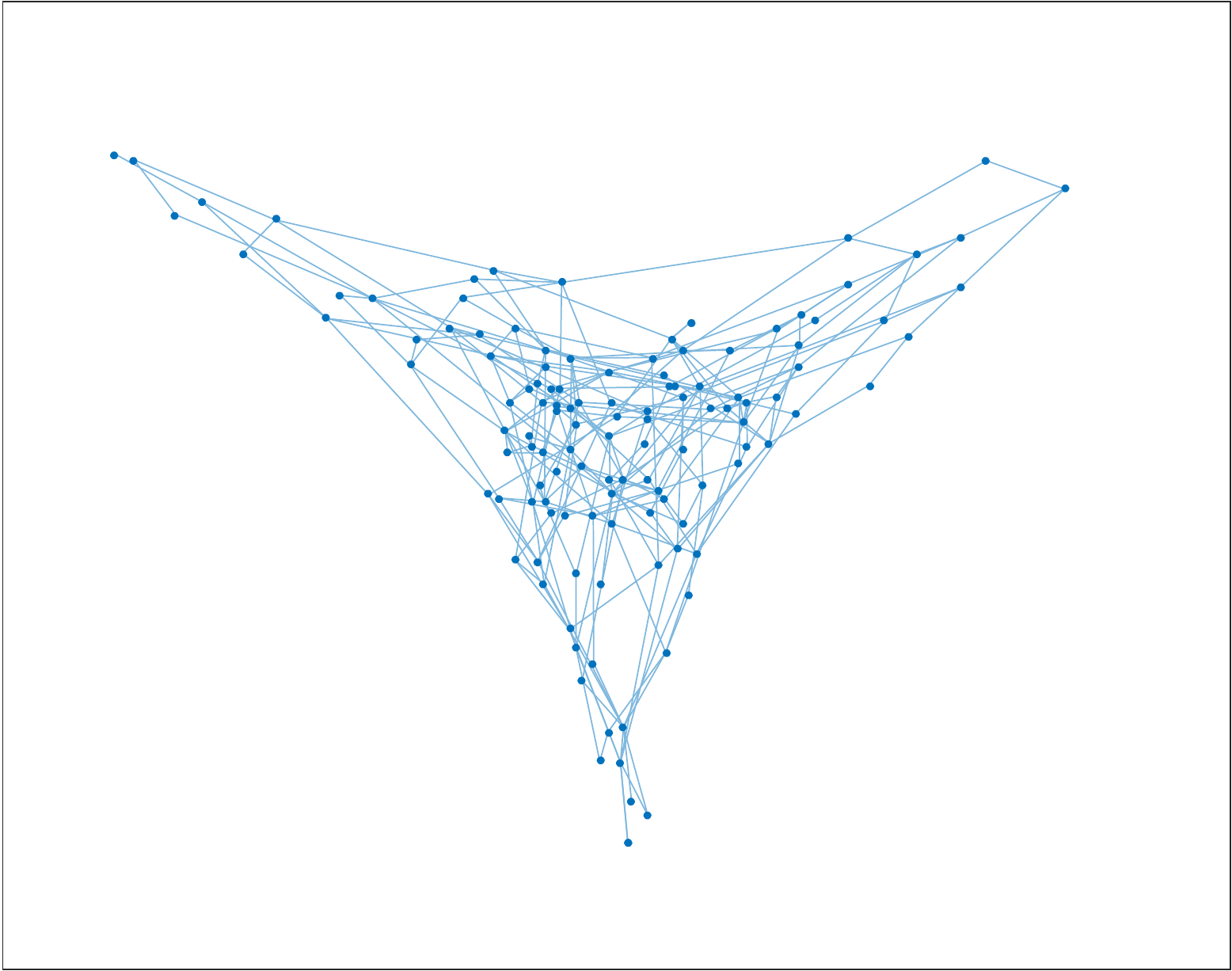}
					\caption{Arbitrary unweighted graph with 125 nodes and 250 edges}
					\label{G125}
				\end{figure}
				
				\vspace{0.1cm}
				\begin{example}\label{Ex:addingLinks}
					Consider the arbitrary network (\ref{eq:system}) with 125 nodes and initially 250 unweighted links given by Figure \ref{G125} in presence of $\tau = 0.017$ delay. We design an optimal topology for the network using SDP relaxation, Simple Greedy, and Random Greedy. Then we compare them with the hard limit to check how close each method can get to the theoretical lower bound of the solution. We add 7,500 new links using SDP method and 2439 new links using Simple Greedy algorithm. We see that network's square of $\HH_2$-norm performance is improved by $88.8 \%$, from 29.28 to 3.27 using Simple Greedy. From the result of Theorem \ref{hardLimit}, the value of the hard limit for the performance of the network is $3.237$ and we know that the global optimal for the problem is greater than the hard limit. It should be further noted that there exists only $1.2\%$ difference between hard limit and the new performance of the network generated by our Simple Greedy algorithm and even smaller gap for SDP. We observe that the result of Random Greedy can be different in each run of the algorithm. It is noteworthy that in the case of this example, although the final network generated by Random Greedy and Simple Greedy are very different, eventually the difference between performance of the network generated by them, is not very different, as it can be seen in Figure \ref{GreedySDP}.  Our simulations confirm that our proposed Simple Greedy performs near-optimal for generic time-delay linear consensus networks. In example \ref{simp-rand}, we construct a specific network that by which it is argued that Random Greedy may outperform Simple Greedy by a considerable margin.
				\end{example}

				\begin{figure}[t]
					\psfrag{x}[c][B][0.8]{\footnotesize{Number of new links (Greedy) \textbackslash Total weight of new links(SDP)}}
					\psfrag{z}[c][c]{$\rhoo$}
					\includegraphics[width=0.45\textwidth,trim={-2.5cm -1.2cm 0cm 0}, clip=true]{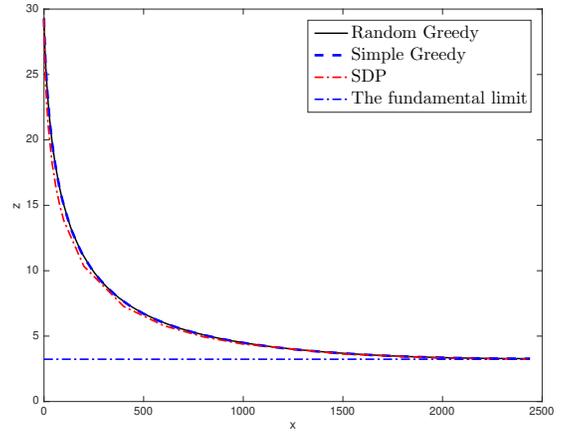}
					\caption{Improving performance of the network given in Figure \ref{G125} by adding new interconnections through SDP and Greedy algorithms }
					\label{GreedySDP}
				\end{figure}

				\begin{example}
					Here we want to evaluate the strength of establishing new interconnections using our greedy algorithm. To that end, we use our greedy algorithm to add edges to a randomly generated graph given in Figure \ref{greedyGraph} which has 10 nodes and 15 edges initially. Here we deal with $\tau = 0.05$ homogeneous time-delay. Moreover, we suppose that set of candidate edges are complement of the set of initial edges. In this example we intend to establish up to 16 new interconnections. As it is depicted in Figure \ref{fig_10}, the algorithm yields extremely good results. {Our simulations results assert that our suggested Simple Greedy provide  near-optimal solution to time-delay linear consensus networks with generic graph topologies.  }
				\end{example}
				
				\begin{figure}[t]
					\centering
					\includegraphics[width=0.4\textwidth,trim={0.5cm 1cm 1cm 1.2cm}, clip=true]{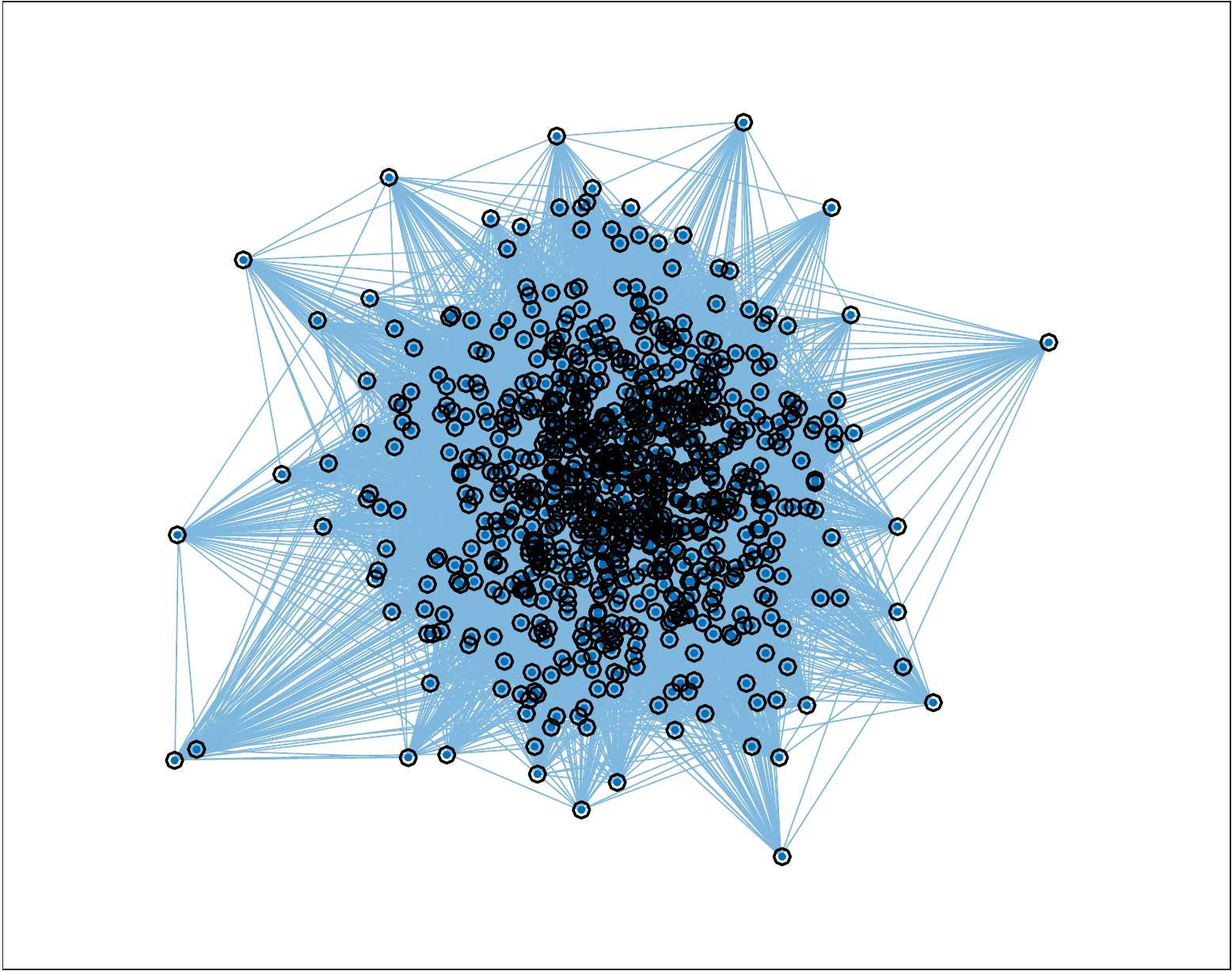}
					\caption{Arbitrary unweighted graph with 800 nodes and 20,000 edges}
					\label{G800}
				\end{figure}
				\begin{figure}[t]
					\psfrag{x}[c][B][0.8]{\footnotesize{Number of removed links by Simple Greedy for Sarsification}}
					\psfrag{z}[c][c]{$\rhoo$}
					\includegraphics[width=0.4\textwidth,trim={-2.5cm -1.2cm 0cm 0}, clip=true]{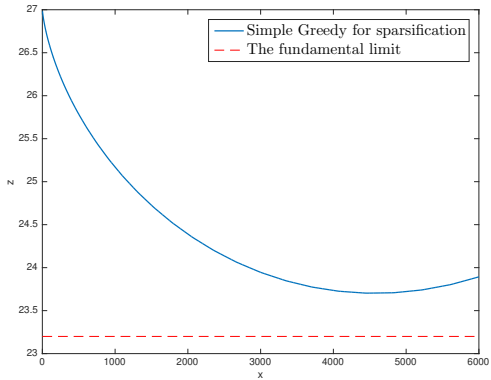}
					\caption{Improving performance of the network given in Figure \ref{G800} by removing interconnections through sparsification}
					\label{GreedySparse}
				\end{figure}

				\begin{example}\label{Ex:removingLinks}
					Let us consider linear consensus network (\ref{eq:system}) with $800$ nodes and initially $2 \times 10^4$ unweighted links given by Figure \ref{G800} in presence of $\tau = 0.019$ delay. Our goal is to remove links from the coupling graph of the network using our sparsification Algorithm 3 in order to improve the performance. We compare the best achieved performance with the hard limit, to see how close we can get to the lower bound of the solution. We removed $6 \times 10^3$ links by executing  Algorithm 3. We observe that the best  performance is achieved by removing $4461$ links and the network's square of $\HH_2$-norm performance is improved by $12\%$, reaching $23.7$, which was initially $27$. Using inequality \eqref{eq:funLimit1}, the hard limit for the performance of the network is $23.2$ and we know that the optimal solution for the problem is greater than or equal to the hard limit. It is noteworthy that there exists only $2.2\%$ difference between hard limit and the best achieved performance.
				\end{example}

				\begin{figure}[t]
					\centering
					\includegraphics[width=0.25\textwidth,trim={0.1cm 1cm 1cm 0.5cm}, clip=true]{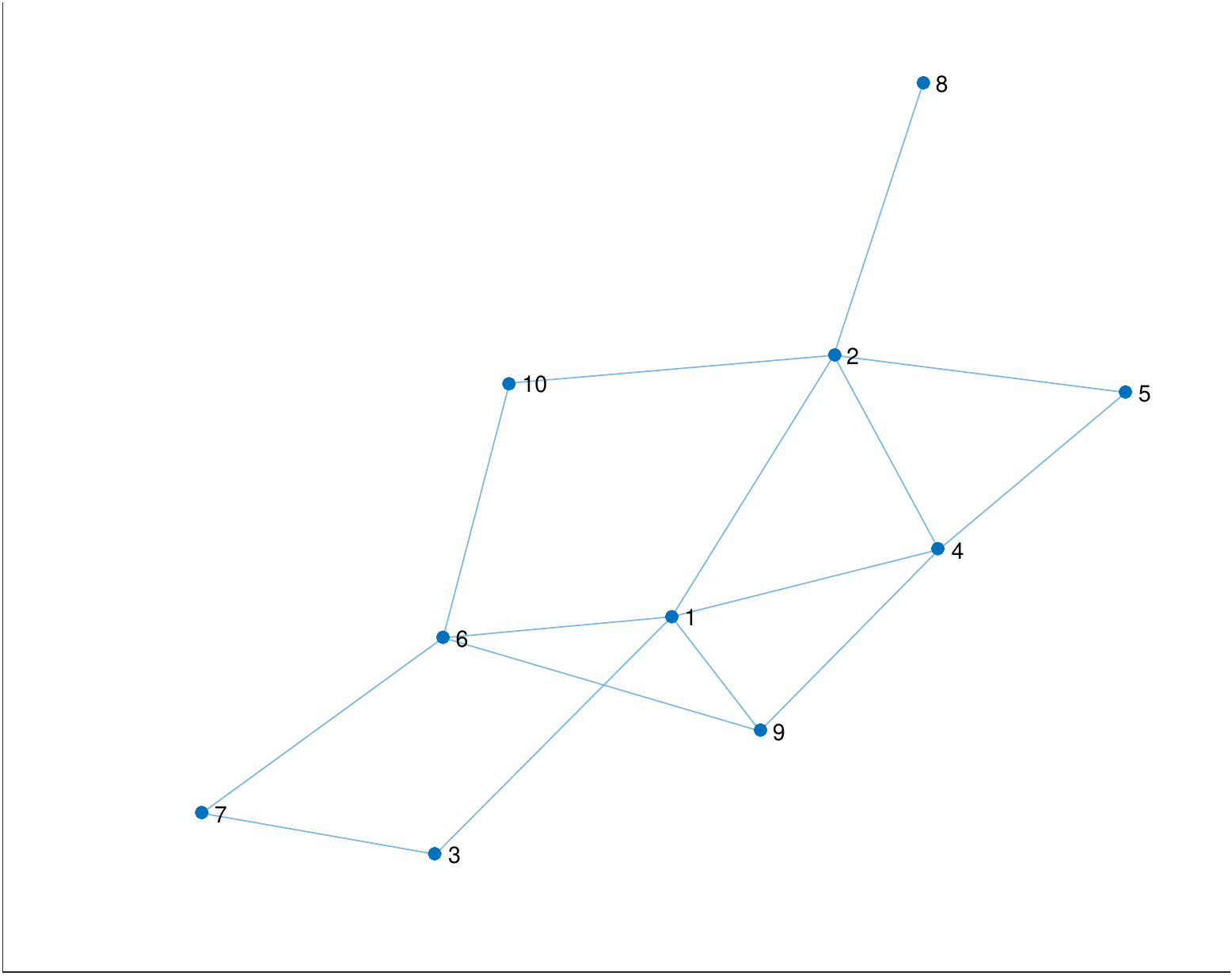}
					\caption{Arbitrary unweighted graph with 10 nodes and 15 edges}
					\label{greedyGraph}
				\end{figure}

				\begin{figure}[t]
					\centering
					\psfrag{x}[c][c]{\footnotesize{Number of new interconnections}}        
					\psfrag{z}[c][c]{$\rhoo$}
					\includegraphics[width=0.33\textwidth,trim={0cm 0 1.6cm 0.7cm}]{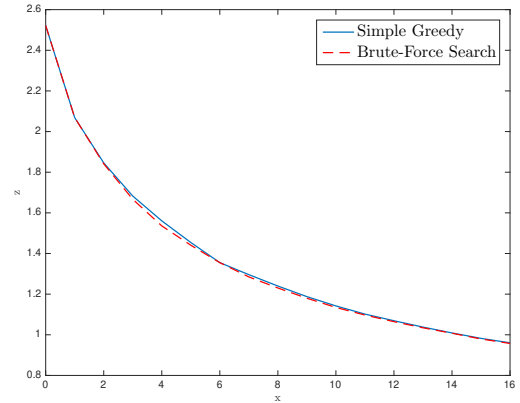}
					\caption{Improving performance of the network given in Figure \ref{greedyGraph} by adding new interconnections}
					\label{fig_10}
				\end{figure}

				\begin{figure}[t]
					\centering{
						\includegraphics[scale=0.38,trim={3cm 2.8cm 1cm 1.5cm},clip]{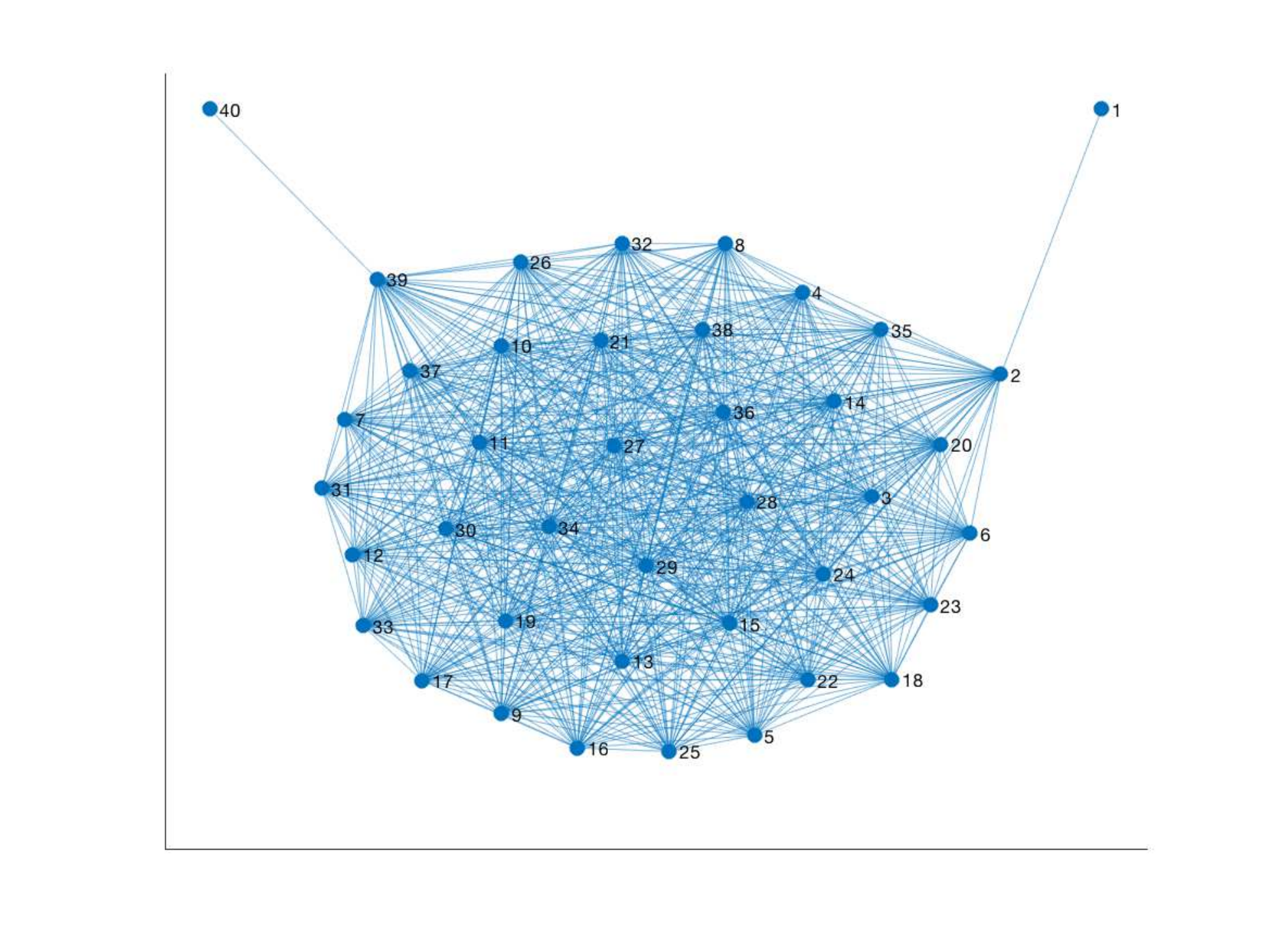}}
					\caption{An arbitrary graph with $40$ nodes with uniform weight $1$. Every subset of nodes $\{2, \dots, 39\}$ forms a clique, i.e., the subgraph induced by them is a complete graph.}
					\label{fig_fail}
				\end{figure}

				\begin{example}\label{simp-rand}
					Suppose that linear consensus network (\ref{eq:system}) with underlying graph in Figure  \ref{fig_fail} and  time-delay $\tau = 1.852 \times 10^{-2}$ is given. Our aim is to show that growing  the network  by Random Greedy sometimes outperforms the Simple Greedy algorithm. The set of candidate links is ${\EE}_c =  {\EE}_{{\mathcal K}_n} \setminus {\EE}$, where ${\EE}_{{\mathcal K}_n}$ is the set of all edges in complete graph on $n$ nodes ${\mathcal{K}}_{\N}$. Let us set $\varpi(e_1) = 40$ for $e_1 = \{1,40\}$ and  $\varpi(e) = 1$ for all other edges $e \in {\EE}_c \setminus \{e_1\}$. The performance measure of the initial network is $2.0686$. The value of the performance measure for the resulting network from Algorithm 1 is reduced to $1.5195$, while by applying Algorithm 2, the performance measure can be reduced up to $1.1038$ with high probability. In this specific example, we observe that the Simple Greedy improves the performance $47\%$ less than the Random Greedy. We would like to mention that if we set $\varpi(e_1)=1$, the Simple Greedy provides us with the optimal solution. 
				\end{example}

				\section{Discussion and Conclusion}
				We studied $\HH_2$-norm performance of noisy  first-order time-delay linear consensus networks from a spectral graph theoretical point of view. It is shown that this measure is convex with respect to weights and  increasing with respect to time-delay. { We propose  low-complexity methods to improve the performance of such networks, where the fastest of which has cubic time algorithm capable of generating sparse solutions. The design problems discussed in this paper can also be formulated as SDP problems, where solving such convex problem requires costly $\mathcal{O}(n^3+|\EE|^3)$ operations for each iteration.} This is why we have favored greedy algorithms to design time-delay linear consensus networks that offer significantly lower time-complexities.  
				
				The focus of this paper has been on time-delay linear consensus networks with state-space matrices $(-L,I,C,0)$, where $C$ is an arbitrary output matrix that are orthogonal to the vector of all ones.  Our methodology can also handle time-delay networks with state-space matrices $(-L,B,M_n,0)$. Using duality of controllability and observability. It is straightforward to show that $\mathcal{H}_2$-norm of the following three networks with state-space matrices $(-L,B,M_n,0)$, $(-L,M_n,B^{\T},0)$, and $(-L,I,B^{\T},0)$  are equal and they are bounded if vector of all ones is in the left nullspace of matrix $B$ and the time-delay is less than the time-delay margin.
				
				In this paper, we assumed that time-delay is uniform across the network. Performance of consensus networks with non-uniform delay is a possible direction to generalize our results.

				\section*{Appendix}
				
				The following definitions and results are used in our proofs in
				this appendix.

				\begin{definition}\label{DelayLyap}
					For a linear delayed system $\hat{G}$ with state-space representation:
					\begin{align*}\hat{G}:
					\begin{cases}
					\dot{\psi}(t)&=A \psi(t-\tau) +Bu(t),\\
					\phi(t)&=C \psi(t),
					\end{cases}
					\end{align*}
					where $u(t)$ is control input of the system. Let ${\Psi(t): \mathbb{R}_{+} \rightarrow \mathbb{R}^{n \times n}}$ be the fundamental matrix of the system with $\Psi(0)=I_{\N}$, then,
					\begin{align*}
					U(t) := \int_{s=0}^{+\infty}{\Psi^{\T}(s)C^{\T}C\Psi(s+t) \: ds}
					\end{align*}
					is called the delay Lyapunov matrix.
					Although there are a few definitions for delay Lyapunov matrix $U(t)$; here we have put energy functional definition into use. 
				\end{definition}

				\begin{theorem}[\cite{Jarlebring:2011}]\label{DelayH2}
					Suppose the time-delay system $\hat{G}$ is exponentially stable, then ${\HH}_2$-norm is
					\begin{align}\label{eq:H2normDelaySystem}
					\|\hat{G}\|_{{\HH}_2}^2&=\Tr\big[B^{\T} U(0) B \big]\\
					&=\Tr\big[C V(0) C^{\T} \big],
					\end{align}
					where $U(t)$, $V(t)$ are the unique solution to the delay Lyapunov equation:
					\begin{subequations}\label{eq:delayLyapunovMatrixeq1}
						\begin{align}
						\dot{U}(t)&=U(t-\tau)A \text{ for } t \in [0,h],\\\label{symmetry}
						U(t)&=U^{\T}(-t), \\
						-C^{\T}C&=U(-\tau)A+A^{\T}U^{\T}(-\tau),
						\end{align}
					\end{subequations}
					and its dual equation:
					\begin{subequations}
						\begin{align*}
						\dot{V}(t)&=V(t-\tau)A^{\T}\text{ for } t \in [0,h],\\
						V(t)&=V^{\T}(-t), \\
						-BB^{\T}&=V(-\tau)A^{\T}+AV^{\T}(-\tau).
						\end{align*}
					\end{subequations}
					
				\end{theorem}
				
				\begin{proof}
					The idea of proof is by substituting $y(t)=C\Psi(t)B$ into Definition (\ref{eq:coherenceDef}). Further, by substituting $U(t)$ from Definition (\ref{DelayLyap}) into (\ref{eq:delayLyapunovMatrixeq1}), (\ref{eq:delayLyapunovMatrixeq1}) can be verified. See \cite{Jarlebring:2011},\cite{Kharitonov:2006} for more details.
				\end{proof}

				\begin{lemma}\label{lemma2}
					Let $\hat{U},\hat{V} : [0, \tau] \rightarrow \mathbb{R}^{n\times n}$ be solution of the following problem:
					\begin{align}\label{vectorProb}
					\frac{d}{dt}\begin{bmatrix}
					\Vect(\hat{U}(t))\\
					\Vect(\hat{V}(t))
					\end{bmatrix}
					=\hat{A}
					\begin{bmatrix}
					\Vect(\hat{U}(t))\\
					\Vect(\hat{V}(t))
					\end{bmatrix}
					\end{align}
					and satisfy boundary condition
					\begin{align}\label{vectorProbBound}
					M
					\begin{bmatrix}
					\Vect(\hat{U}(0))\\
					\Vect(\hat{V}(0))
					\end{bmatrix}
					+N
					\begin{bmatrix}
					\Vect(\hat{U}(\tau))\\
					\Vect(\hat{V}(\tau))
					\end{bmatrix}
					=
					\begin{bmatrix}
					-\Vect(C^{\T}C)\\
					0_{n^2}
					\end{bmatrix},
					\end{align}
					where 
					\[
					M=
					\begin{bmatrix}
					0_{n^2 \times n^2} & A^{\T} \otimes I_n\\
					I_{n^2} & 0_{n^2 \times n^2}
					\end{bmatrix},~~~~
					N=
					\begin{bmatrix}
					I_n  \otimes A^{\T} & 0_{n^2 \times n^2} & \\
					0_{n^2 \times n^2} & -I_{n^2} 
					\end{bmatrix},
					\]
					\begin{align*}
					\hat{A}=
					\begin{bmatrix}
					0_{n^2 \times n^2} & A^{\T} \otimes I_n\\
					-I_{n}\otimes A^{\T} & 0_{n^2 \times n^2}
					\end{bmatrix}.
					\end{align*}
					Then, 
					\begin{align*}
					U(t) = \frac{1}{2}
					\begin{cases}
					\hat{U}(t) + \hat{V}(\tau-t)\text{\:\: for \:\:}t \in [0, \tau],\\
					\hat{U}^{\T}(-t) + \hat{V}(\tau+t)\text{\:\: for \:\:}t \in [-\tau, 0),
					\end{cases}
					\end{align*}
					solves (\ref{eq:delayLyapunovMatrixeq1}).
				\end{lemma}
				\begin{proof}
					Problem (\ref{vectorProb}) under condition (\ref{vectorProbBound}) is a vectorization of the problem given in (\cite[Problem 8]{Kharitonov:2006}) which solves (\ref{eq:delayLyapunovMatrixeq1}). See \cite{Kharitonov:2006}, \cite[6.45 \& 6.40]{Plischke:2005} for more details.
				\end{proof}
				
				\begin{definition}
					A function $h$ on the set of $\N \times \N$ Hermitian matrices is called unitary invariant, provided
					\begin{eqnarray*}
						h(U^{-1}YU)=h(Y),
					\end{eqnarray*}
					for any unitary matrix $U \in \mathbb{R}^{\N \times \N}$.
				\end{definition}

				\begin{definition}
					A function $f: \mathbb{R}^{\N} \to \mathbb{R}$ is symmetric if for all permutation matrices $P \in \mathbb{R}^{\N \times \N}$,
					\begin{eqnarray*}
						f(Pz)=f(z).
					\end{eqnarray*}
				\end{definition}

				\begin{theorem}[\cite{Varberg:1973}]\label{Schur_H}
					A convex function $f$ is Schur-convex if it is symmetric.
				\end{theorem}
				
				\begin{proof}
					Since $f$ is convex, there is a permutation matrix $P^*$ such that,
					\begin{align*}
					f(Dz)\leq f(P^*z).
					\end{align*}
					for all doubly stochastic matrices $D$\cite[Corollary 8.7.4]{horn2012matrix}. Further, by considering that $f$ is symmetric,
					
					\begin{align*}
					f(Dz)\leq f(z),
					\end{align*}
					implying Schur-convexity of $f$. 
					For more details see \cite{Varberg:1973}.
				\end{proof}

				\begin{lemma}\label{H2normOfScalerDelayedSystem}
					For $0 \leq \tau<\frac{\pi}{2\lambda_{\N}}$, the integral
					\begin{eqnarray}\label{eq:integralComplex}
					\int_{-\infty}^{+\infty}{\frac{1}{\big(j\omega+\lambda e^{-j\tau\omega}\big)\big({-j\omega+\lambda e^{j\tau\omega}}\big)} ~d\omega},
					\end{eqnarray}
					is well defined and equals
					\begin{eqnarray*}
						\frac{\pi \cos(\lambda \tau)} {\lambda-\lambda \sin(\lambda \tau)}.
					\end{eqnarray*}
				\end{lemma}
				
				\begin{proof}
					Integral (\ref{eq:integralComplex}), is square of ${\HH}_2$-norm of a system with the following transfer function
					\begin{eqnarray*}
						\hat{G}(s)=\frac{\sqrt{2\pi}}{s+\lambda e^{-\tau s}},
					\end{eqnarray*}
					for which we have the following state space representation
					\begin{eqnarray}\label{ssRep}
					\begin{cases}
					\dot{\psi}(t)&=-\lambda \psi(t-\tau) + \sqrt{2\pi} u(t),\\
					\phi(t)&=  \psi(t).
					\end{cases}
					\end{eqnarray}
					To find value of delay Lyapunov matrix, we apply Lemma \ref{lemma2} to the system (\ref{ssRep}), and we get
					\begin{eqnarray}\label{eq:U0}
					U(0)=\frac{\cos(\lambda \tau)}{2\big(\lambda-\lambda\sin(\lambda\tau)\big)}.
					\end{eqnarray}
					Moreover, by substituting $U(0)$ from equality (\ref{eq:U0}) into equation (\ref{eq:H2normDelaySystem}), we conclude that
					\begin{eqnarray}\label{H2Ghat}
					\|\hat{G}\|_{\HH_2}^2 =\frac{\pi \cos(\lambda \tau)} {\lambda-\lambda \sin(\lambda \tau)}.
					\end{eqnarray}
				\end{proof}
				
				\begin{lemma}\cite[2.3]{zhang2002spectral}\label{specRad}
					For an unweighted graph $\mathcal{G}$ with $\N$ nodes and at least one edge, we have
					$$\lambda_n(\mathcal{G}) \geq d_{\mbox{max}} +1,$$
					furthermore, if $\mathcal{G}$ is connected equality holds if and only if 
					$$d_{\mbox{max}} =\N -1.$$
				\end{lemma}


				\bibliographystyle{IEEEtr}
				\bibliography{IEEE_TAC-August1-2018}

				\begin{IEEEbiography}[{\includegraphics[width=1in,height=1.25in,clip,keepaspectratio]{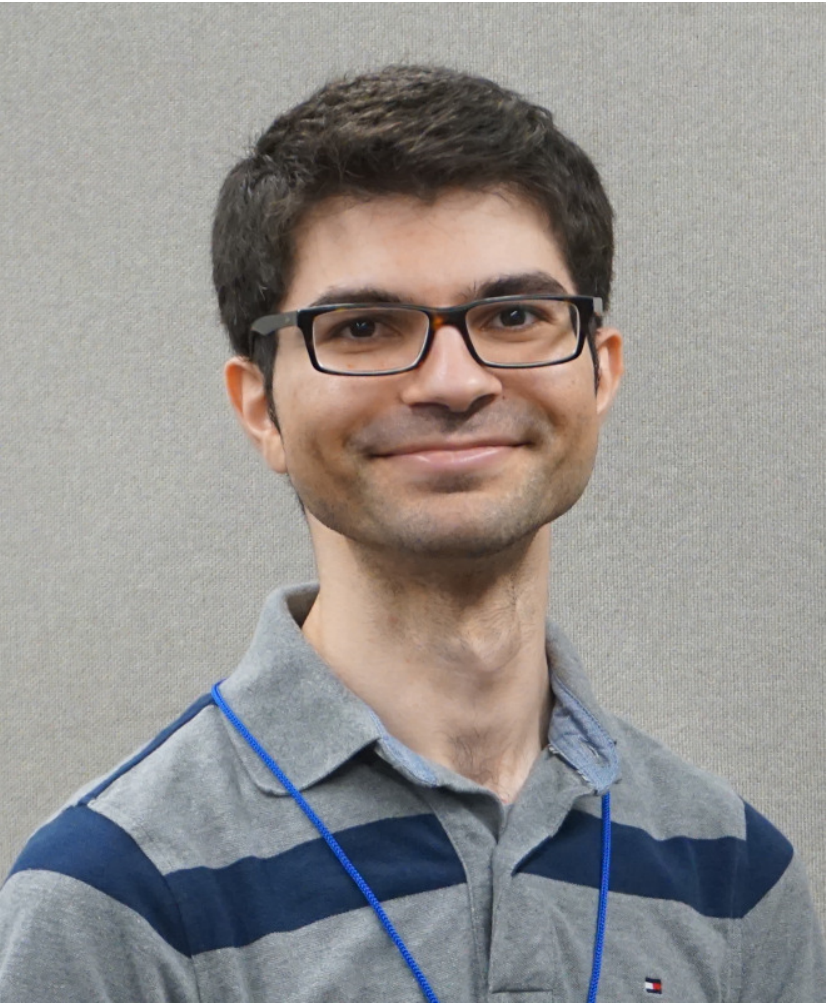}}]{Yaser Ghaedsharaf}
					 received his B.Sc. degree in Mechanical Engineering from Sharif University of Technology in 2013. He is currently pursuing a Ph.D. in the Department of Mechanical Engineering \& Mechanics at Lehigh University. He is the Runner-Up for the Best Student Paper Award in the 6th IFAC Workshop on Distributed Estimation and Control in Networked Systems in 2016 . Furthermore, he was a recipient of the Rossin College of Engineering Doctoral Fellowship in 2016, the RCEAS fellowship award in 2018, and the Mountaintop Research Fellowship in 2018. His research interests include machine learning, analysis and optimal design of networked control systems with applications in distributed control and cyber-physical systems, and robotics.
				\end{IEEEbiography}

				\vspace{-0.8cm} 
				\begin{IEEEbiography}[{\includegraphics[width=1in,height=1.25in,clip,keepaspectratio, trim= 0 0 0 0 ]{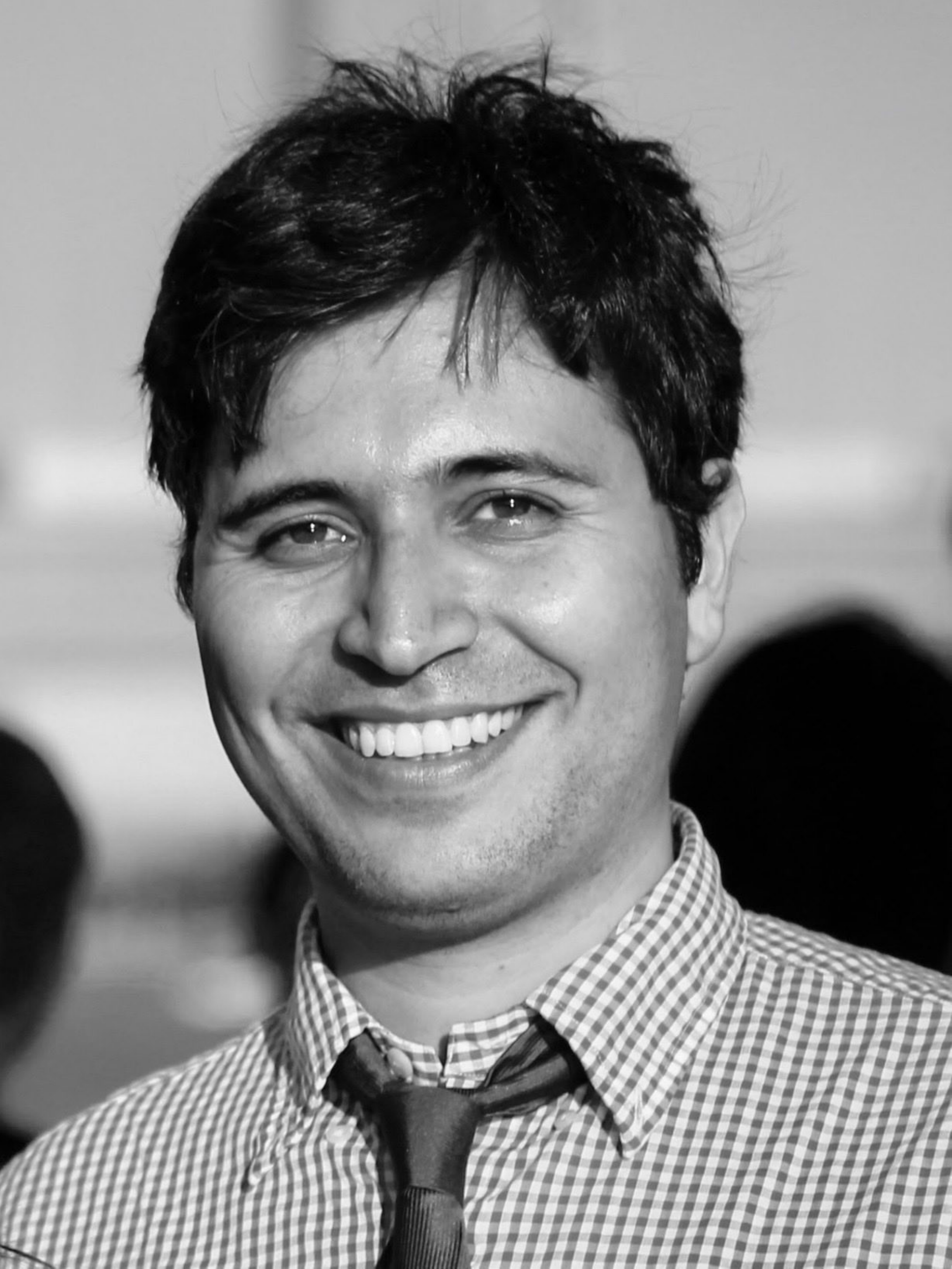}}]{Milad Siami} 
					(S'12-M'18) received his dual B.Sc. degrees in electrical engineering and pure mathematics from Sharif University of Technology in 2009, M.Sc. degree in electrical engineering from Sharif University of Technology in 2011, and M.Sc. and Ph.D. degrees in mechanical engineering from Lehigh University in 2014 and 2017 respectively. 
					From 2009 to 2010, he was a research student at the Department of Mechanical and Environmental Informatics at the Tokyo Institute of Technology, Tokyo, Japan. He is currently a postdoctoral associate in the Institute for Data, Systems, and Society at MIT. His research interests include distributed control systems, distributed optimization, and applications of fractional calculus in engineering. 
					Dr. Siami received a Gold Medal of National Mathematics Olympiad, Iran (2003) and the Best Student Paper Award at the 5th IFAC Workshop on Distributed Estimation and Control in Networked Systems (2015). Moreover, he was awarded  RCEAS Fellowship (2012), Byllesby Fellowship (2013), Rossin College Doctoral Fellowship (2015), and Graduate Student Merit Award (2016) at Lehigh University.
				\end{IEEEbiography}
				\vspace{-0.8cm}

\begin{IEEEbiography}[{\includegraphics[width=1in,height=1.25in,clip,keepaspectratio]{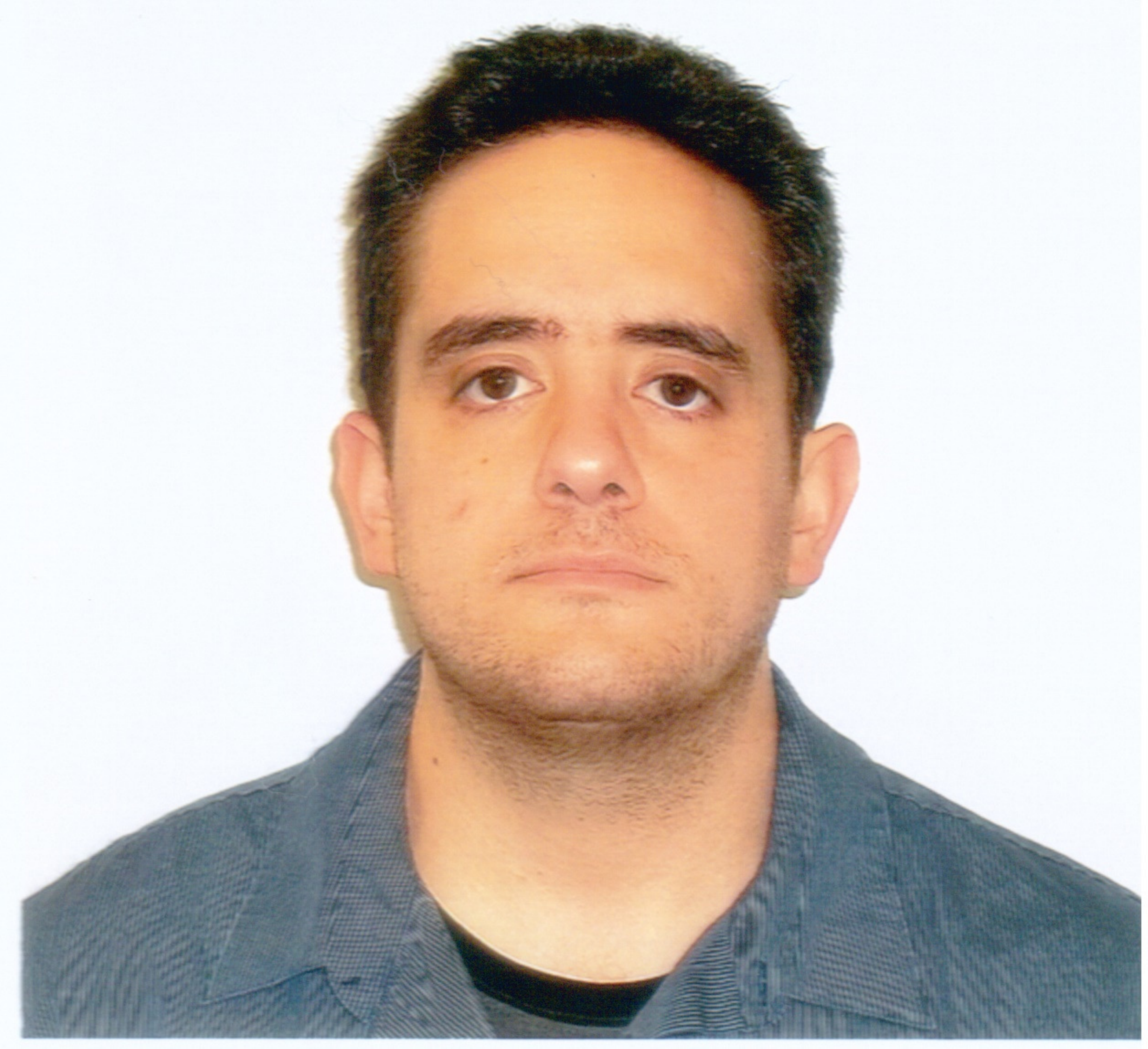}}]{Christoforos Somarakis}
					 received the B.S. degree in
Electrical Engineering from the National Technical University of Athens, Athens, Greece, in 2007 
and and the M.S. and Ph.D. degrees in applied mathematics from the University of Maryland at College Park, in 2012 and 2015, respectively. He is currently a Research Scientist in the Department of Mechanical Engineering 
and Mechanics at Lehigh University. His research interests include analysis and optimal design of networked control systems with applications in distributed control and cyber-physical systems.
				\end{IEEEbiography}

				\begin{IEEEbiography}[{\includegraphics[width=1in,height=1.25in,clip,keepaspectratio]{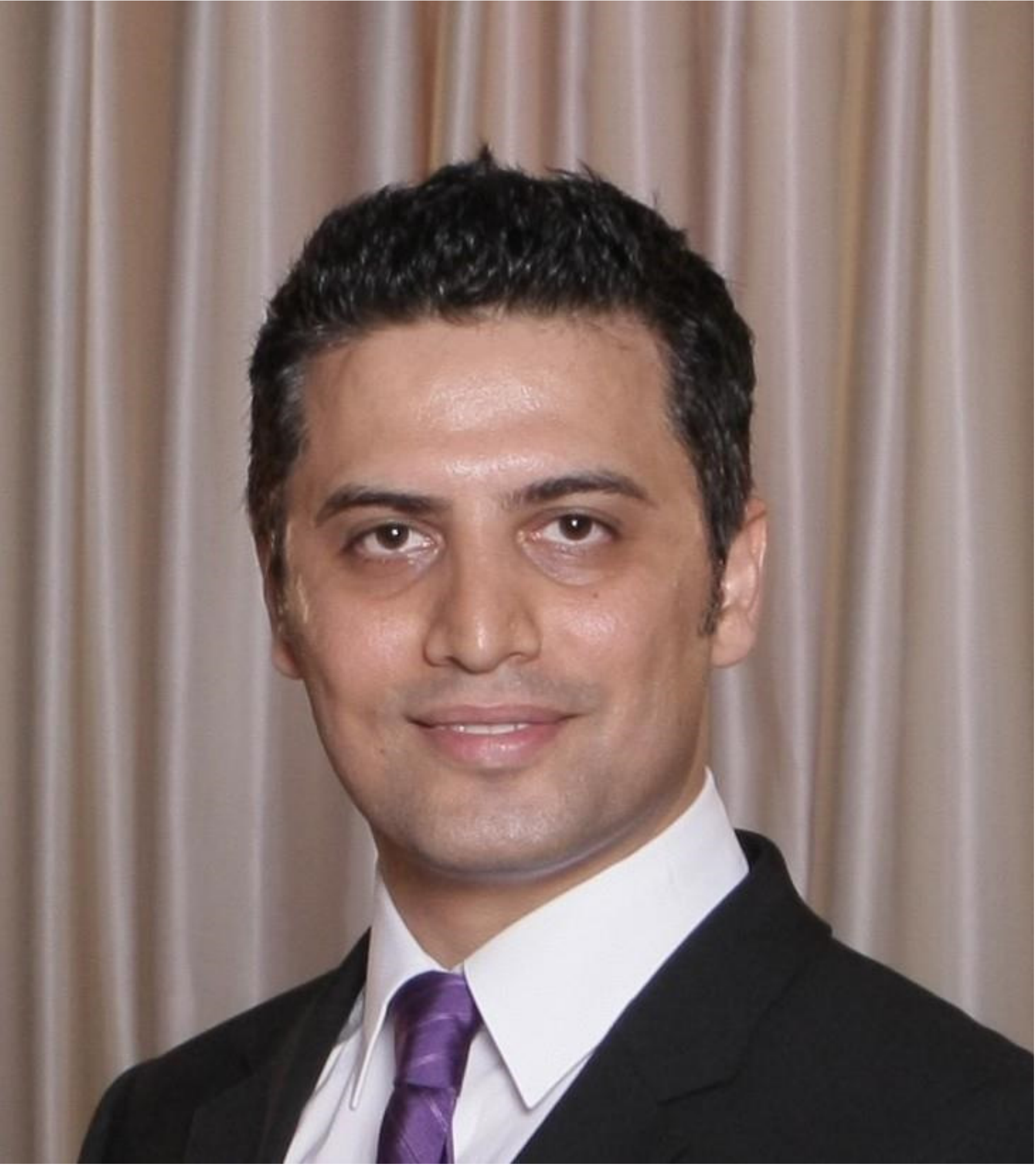}}]{Nader Motee}
					(S'99-M'08-SM'13) received his B.Sc. degree in Electrical Engineering from Sharif University of Technology in 2000, M.Sc. and Ph.D. degrees from University of Pennsylvania in Electrical and Systems Engineering in 2006 and 2007 respectively. From 2008 to 2011, he was a postdoctoral scholar in the Control and Dynamical Systems Department at Caltech. He is currently an  Associate Professor in the Department of Mechanical Engineering and Mechanics at Lehigh University. His current research area is distributed dynamical and control systems with particular focus on issues related to sparsity, performance, and robustness. He is a past recipient of several awards including the 2008 AACC Hugo Schuck best paper award, the 2007 ACC best student paper award, the 2008 Joseph and Rosaline Wolf best thesis award, a 2013 Air Force Office of Scientific Research  Young Investigator Program award (AFOSR YIP), a 2015 NSF Faculty Early Career Development (CAREER) award, and a 2016 Office of Naval Research Young Investigator Program award (ONR YIP).
				\end{IEEEbiography}
				
			\end{document}